\documentclass[letterpaper,11pt]{article}
\pdfoutput=1

\usepackage{amsmath}
\usepackage{mathtools}
\usepackage{amssymb}
\usepackage{graphicx}
\usepackage{caption}
\usepackage{geometry}
\usepackage{enumitem}
\usepackage{titlesec}
\usepackage{amsthm}
\usepackage{thmtools}
\usepackage[hidelinks]{hyperref}
\usepackage{mathrsfs}
\usepackage{titling}

\numberwithin{equation}{section}

\declaretheoremstyle[
    spaceabove=\topsep, spacebelow=\topsep,
    bodyfont=\normalfont
]{mydefstyle}

\declaretheorem[name=Theorem,
    refname={Theorem,Theorems},
    numberwithin=section]{thm}
\declaretheorem[name=Lemma,
    refname={Lemma,Lemmas},
    sibling=thm]{lemma}
\declaretheorem[name=Corollary,
    refname={Corollary,Corollaries},
    sibling=thm]{corollary}
\declaretheorem[name=Definition,
    style=mydefstyle,
    numberwithin=section,
    refname={Definition,Definitions}]{mydef}

\newcommand{\sig}{\mathsf{S}}
\newcommand{\supp}{\mathrm{supp}}
\newcommand{\arity}{\mathrm{arity}}
\newcommand{\inp}[1]{\mathrm{inp}(#1)}
\newcommand{\innerprod}[1]{\left<#1\right>}

\newcommand{\myvert}{\,|\,}

\newcommand{\cP}{\mathrm{P}}
\newcommand{\cNP}{\mathrm{NP}}
\newcommand{\sP}{\#\mathrm{P}}
\newcommand{\redT}{\leq_\mathsf{T}}
\newcommand{\eqredT}{\equiv_\mathsf{T}}

\newcommand{\Holant}{\mathrm{Holant}}
\newcommand{\CSP}{\#\mathrm{CSP}}
\newcommand{\atype}{\mathscr{A}}
\newcommand{\ptype}{\mathscr{P}}
\newcommand{\ttype}{\mathscr{T}}

\title{The Complexity of Holant Problems over Boolean Domain \\with Non-negative Weights}
\date{}
\author{Jiabao Lin\\ Peking University\\ joblin@pku.edu.cn \and Hanpin Wang\\ Peking University\\ whpxhy@pku.edu.cn}

\begin{document}

\begin{titlingpage}
    \maketitle
    \begin{abstract}
    \normalsize
        Holant problem is a general framework to study the computational complexity of counting problems. We prove a complexity dichotomy theorem for Holant problems over the Boolean domain with non-negative weights. It is the first complete Holant dichotomy where constraint functions are not necessarily symmetric.

        Holant problems are indeed read-twice \#CSPs. Intuitively, some \#CSPs that are \#P-hard become tractable when restricted to read-twice instances. To capture them, we introduce the Block-rank-one condition. It turns out that the condition leads to a clear separation. If a function set $\mathcal{F}$ satisfies the condition, then $\mathcal{F}$ is of affine type or product type. Otherwise (a) $\Holant(\mathcal{F})$ is \#P-hard; or (b) every function in $\mathcal{F}$ is a tensor product of functions of arity at most 2; or (c) $\mathcal{F}$ is transformable to a product type by some real orthogonal matrix. Holographic transformations play an important role in both the hardness proof and the characterization of tractability.
    \end{abstract}
\end{titlingpage}

\section{Introduction}

There has been considerable interest in several frameworks to study the complexity of counting problems. One natural framework is the counting Constraint Satisfaction Problem (\#CSP) \cite{CH1996,BD2007,DGJ2009,BDGJR2009,DR2010,BDGJJR2012,CCL2011,CC2012,Bulatov2013}. Another is Graph Homomorphism (GH) \cite{Lovasz1967,HN1990,DG2000,BG2005,DGP2007,GGJT2010,CC2010,CCL2013}, which can be seen as a special case of \#CSP. Such frameworks express a large class of counting problems in the Sum-of-Product form. It is known that if $\cP\neq \cNP$, then there exists a problem that is neither in P nor NP-complete \cite{Ladner1975}. And there is an analogue of Ladner's Theorem for the class \#P. However, for these frameworks, various beautiful dichotomy theorems have been proved, classifying all problems in the broad class into those which are computable in polynomial time (in P) and those which are \#P-hard. A natural question is: For how broad a class of counting problems can one prove a dichotomy theorem?

While GH can express many interesting graph parameters, Freedman, Lov{\'a}sz and Schrijver \cite{FLS2007} showed that the number of perfect matchings of a graph cannot be represented as a homomorphism function. Inspired by holographic algorithms \cite{Valiant2008,Valiant2006}, Cai, Lu and Xia \cite{CLX2009} proposed a more refined framework called Holant Problems. Here we give a brief introduction. In this paper, constraint functions are defined over the Boolean domain, if not specified. Let $\mathcal{F}$ denote a set of algebraic complex-valued functions. A \emph{signature grid} $\Omega$ is a tuple $(G,\mathcal{F},\pi)$ where $G=(V,E)$ is an undirected graph, and $\pi$ is a map that maps each vertex $v\in V$ to some function $f_v\in\mathcal{F}$ and its incident edges $E(v)$ to the input variables of $f_v$. The counting problem on $\Omega$ is to compute
\begin{equation*}
    \Holant_\Omega = \sum_{\sigma:E\to\{0,1\}}\prod_{v\in V}f_v(\sigma |_{E(v)}),
\end{equation*}
where $\sigma |_{E(v)}$ is the restriction of $\sigma$ to $E(v)$. All such signature grids constitute the set of instances of the problem $\Holant(\mathcal{F})$. For example, consider the problem of counting perfect matchings (\#PM) on graph $G$. In a perfect matching, every vertex is saturated by exactly one edge. Such constraint on a vertex of degree $n$ can be expressed as an \textsc{Exact}-\textsc{One} function $f:\{0,1\}^n\to\{0,1\}$, which takes the value 1 if and only if its input has Hamming weight 1. If every vertex is assigned such a function, then the value $\Holant_\Omega$ is exactly the number of perfect matchings. Let $\mathcal{F}$ denote the set of all \textsc{Exact}-\textsc{One} functions, then $\Holant(\mathcal{F})$ represents the problem \#PM.

The Holant framework is general enough: \#CSPs can be viewed as special Holant problems where all equality functions are available \cite{CLX2009}. However, the very generality makes it more difficult to prove a dichotomy. A function is \emph{symmetric} if the function values only depend on the Hamming weights of inputs, like the \textsc{Exact}-\textsc{One} functions. Satisfactory progress has been made in the complexity classification of Holant problems specified by sets of symmetric functions \cite{CHL2012,HL2012,GLV2013,CGW2013,CFGW2015}. And in the process, some unexpected tractable classes were discovered. They give many deep insights into both tractability and hardness.

It still remains open whether a complete dichotomy exists, since the definition of Holant problems does not require that constraint functions be symmetric. Such restriction is stringent and generally it is not imposed in \#CSP. Cai, Lu and Xia \cite{CLX2011b} proved a dichotomy without symmetry for a special family of Holant problems, called $\Holant^*$, where all unary functions are assumed to be available. But without this assumption, as in \cite{CGW2013}, more tractable classes will be released, which makes the hardness proof very different.

We prove a dichotomy theorem for Holant problems with non-negative algebraic real weights. It is the first complete Holant dichotomy where constraint functions are not necessarily symmetric and no auxiliary function is assumed to be available. This generalizes the results on Boolean \#CSP in $\cite{CH1996,DGJ2009}$, and the dichotomies in \cite{HL2012,CGW2013} restricted to non-negative case. Our proof starts with an infinitary condition, but finally obtains an explicit criterion (\autoref{thm:dichotomy}).
%
%

%
%


A simple observation is that, Holant problems are indeed \emph{read-twice} \#CSPs where every variable in an instance appears exactly twice (see \autoref{subsec:CSP}). Intuitively, some \#CSPs that are \#P-hard become tractable when restricted to read-twice instances. To capture them, we need insights into what makes a problem hard in \#CSP. Inspired by dichotomy theorems over general domains \cite{BG2005,DR2013,CCL2011,CC2012}, we introduce the Block-rank-one condition for Holant problems (see \autoref{subsec:block-rank-one}). It is known that non-block-rank-one structures imply  hardness in \#CSP. So our condition is necessary for tractability since it is imposed on the functions defined by read-twice instances. Surprisingly, on the Boolean domain, the Block-rank-one condition is also \emph{sufficient} and leads to a clear separation:
\begin{enumerate}[label=\Roman*.,align=left,leftmargin=0pt,listparindent=\parindent,labelwidth=0pt,itemindent=!]
  \item \textbf{Function set $\mathcal{F}$ satisfies the condition}. Then $\CSP(\mathcal{F})$ is in P, and hence its subproblem $\Holant(\mathcal{F})$ is also in P.
  \item \textbf{Function set $\mathcal{F}$ violates the condition}. Then (a) $\Holant(\mathcal{F})$ is \#P-hard or (b) $\CSP(\mathcal{F})$ is \#P-hard but $\Holant(\mathcal{F})$ is tractable.
\end{enumerate}

First we discuss Part II. We can prove \#P-hardness directly, or further induce an orthogonal holographic transformation. After performing the transformation, we have to handle real-valued functions. Luckily, we can even prove a dichotomy theorem for a family of \emph{complex-valued} Holant problems (\autoref{thm:Holant-nontrivial-eq}). And towards this theorem, we prove a lemma (\autoref{lemma:multiple-to-single}) on how to ``extract'' a function from its tensor powers. The proof is non-constructive and the idea can simplify some existing proofs. For example,
it can be shown directly that the two problems $\CSP^d(\mathcal{F}\cup\{[1,0]^{\otimes d},[0,1]^{\otimes d}\})$ and $\CSP^d(\mathcal{F}\cup\{[1,0],[0,1]\})$ in \cite{HL2012} are equivalent under polynomial-time Turing reduction.

%

Now consider Part I. It can be derived that $\mathcal{F}$ is of affine type or $\mathcal{F}$ is of product type, exactly the criterion given by Dyer, Goldberg and Jerrum \cite{DGJ2009}. Dichotomies for \#CSP over general domains \cite{Bulatov2013,DR2013,BDGJJR2012,CCL2011} are very different from those over the Boolean domain \cite{CH1996,DGJ2009}. Our proof builds a connection between them.

The Block-rank-one condition is a little conceptual. To obtain the structure of $\mathcal{F}$, we introduce an equivalent notion, called \emph{balance}, for Holant problems (see \autoref{subsec:balance}). The equivalence is simply built on the concept of \emph{vector representation} in \cite{CCL2011}, which was used to design a polynomial-time algorithm for \#CSP. Back to non-negative \#CSP, we find that actually the notions of weak balance and balance (different from our version for Holant) in \cite{CCL2011} are equivalent, without assuming $\mathrm{FP}\neq \sP$. Therefore, to decide the complexity of a problem $\CSP(\mathcal{F})$, we only need to decide whether $\mathcal{F}$ is of weak balance.

The remainder of this paper is organized as follows. Preliminary definitions and notations are given in \autoref{sec:preliminaries}. In \autoref{sec:decomposition}, we show that, given a function $F=f\otimes g$, under certain conditions we may assume that the component $f$ is freely available. This will be technically useful in later proofs. In \autoref{sec:non-unitary}, we prove \autoref{thm:Holant-nontrivial-eq}. It is an important part of hardness. Some direct applications of this theorem are presented in \autoref{sec:p-transformability}. And in \autoref{sec:special-4}, we consider certain functions of arity $4$ and complete the preparations for Part II. The dichotomy is introduced in \autoref{sec:dichotomy} (\autoref{thm:dichotomy}). In \autoref{subsec:block-rank-one} we define the Block-rank-one condition and finish the proof of Part II. And the remaining two subsections are devoted to Part I. In \autoref{sec:back-to-csp}, we give a simple proof of the equivalence between the two notions defined in \cite{CCL2011}.

\section{Preliminaries}
\label{sec:preliminaries}
\subsection{Functions and Signatures}

Let $\mathbb{C}$ and $\mathbb{R}_+$ denote the set of algebraic complex numbers and the set of algebraic non-negative real numbers, respectively. Throughout this paper, we refer to them simply as complex and non-negative numbers.

The functions we discussed are over the Boolean domain $\{0,1\}$, if not specified.

Given a function $f:\{0,1\}^n\to \mathbb{C}$, we will often write it as a vector of dimension $2^n$ whose entries are the function values, indexed by $\mathbf{x}\in\{0,1\}^n$ lexicographically. This vector is called a \emph{signature}. If the values of an $n$-ary function only depend on the Hamming weights of inputs, then the function is called \emph{symmetric} and can be expressed as $[f_0,f_1,...,f_n]$ where $f_k$ is the function value for inputs of Hamming weight $k$. For example, the ternary logic OR function has the signature $[0,1,1,1]$.

Generally, given a function $f$ of arity $n$, we can express it as a $2^r\times 2^{n-r}$ matrix ($1\leq r\leq n$), denoted by $M_{[r]}(f)$. The rows and columns are indexed by $\mathbf{x}\in\{0,1\}^r$ and $\mathbf{y}\in\{0,1\}^{n-r}$ respectively, and $f(\mathbf{x},\mathbf{y})$ is the $(\mathbf{x},\mathbf{y})^{\mathrm{th}}$ entry of the matrix. And the matrices $\{M_{[r]}(f)\myvert r\in [n]\}$ are called the \emph{signature matrices} of $f$. When the integer $r$ is clear from the context, we simply write $M_f$. For example, given a function $f$ of arity $4$, we often write it as a $4\times 4$ matrix:

\begin{equation*}
   M_f= \begin{bmatrix}
            f_{0000} & f_{0001} & f_{0010} & f_{0011} \\
            f_{0100} & f_{0101} & f_{0110} & f_{0111} \\
            f_{1000} & f_{1001} & f_{1010} & f_{1011} \\
            f_{1100} & f_{1101} & f_{1110} & f_{1111} \\
        \end{bmatrix}.
\end{equation*}

In most cases, if not confused, we identify functions, signatures and signature matrices. But in \autoref{sec:dichotomy}, we shall distinguish a function from its matrix representations.

Let $S_n$ denote the symmetric group on indices $\{1,2,...,n\}$. Given an $n$-ary function $f$ and a permutation $\pi\in S_n$, we define the function $f_\pi(x_1,x_2,...,x_n)=f(x_{\pi(1)},x_{\pi(2)},...,x_{\pi(n)})$.

Let $f:\{0,1\}^r\to \mathbb{C}$ and $g:\{0,1\}^s\to \mathbb{C}$ be two functions, and $r,s\geq 1$. We use $f\otimes g$ to denote their \emph{tensor product}, the function $F$ of arity $r+s$ such that for all $\mathbf{x}\in\{0,1\}^r$ and $\mathbf{y}\in\{0,1\}^s$,
\begin{equation*}
    F(\mathbf{x},\mathbf{y})=f(\mathbf{x})g(\mathbf{y}).
\end{equation*}
Let $h$ be a function and $\pi$ be a permutation such that $h_\pi=f\otimes g$ for some $f$ and $g$. If we do not care about the permutation, we also write $h=f\otimes g$. A function $F$ is \emph{reducible} if $F_\pi$ is a tensor product of two functions for some permutation $\pi$. Otherwise $F$ is called \emph{irreducible}.

A function is called \emph{degenerate} if it is a tensor product of some unary functions. Otherwise we call it \emph{non-degenerate}. In particular, a complex-valued symmetric signature $f=[f_0,f_1,...,f_n]$ is degenerate if and only if $f=[x,y]^{\otimes n}$ for some $x,y\in \mathbb{C}$.

We use $\arity(f)$ to denote the arity of a function $f$.

\subsection{Holant Problems}

Let $\mathcal{F}$ be a (not necessarily finite) set of complex-valued functions. A \emph{signature grid} $\Omega$ is a tuple $(G,\mathcal{F},\pi)$ where $G=(V,E)$ is an undirected graph, and $\pi$ is a map that maps each vertex $v\in V$ to some function $f_v\in\mathcal{F}$ and its incident edges $E(v)$ to the input variables of $f_v$. The counting problem on $\Omega$ is to compute
\begin{equation*}
    \Holant_\Omega = \sum_{\sigma:E\to\{0,1\}}\prod_{v\in V}f_v(\sigma |_{E(v)}),
\end{equation*}
where $\sigma |_{E(v)}$ is the restriction of $\sigma$ to $E(v)$. When $\mathcal{F}$ is fixed, we simply use $(G,\pi)$ to denote a signature grid.

\begin{mydef}[Holant Problems]
    Given a function set $\mathcal{F}$, we define the counting problem $\Holant(\mathcal{F})$:\\
    \emph{Input}: A signature grid $\Omega=(G,\pi)$;\\
    \emph{Output}: $\Holant_\Omega$.
\end{mydef}

Note that the function set $\mathcal{F}$ can be infinite. We say that $\Holant(\mathcal{F})$ is \#P-hard, if there is a finite subset $\mathcal{G}$ of $\mathcal{F}$ such that the problem $\Holant(\mathcal{G})$ is \#P-hard. When $\mathcal{F}$ is infinite, an input instance of  $\Holant(\mathcal{F})$ should include the descriptions of functions that appear in the signature grid.

To introduce the holographic reductions, we define bipartite Holant problems. Let $\Holant(\mathcal{F}\myvert\mathcal{G})$ denote the Holant problem on bipartite graphs $H=(U,V,E)$ where each vertex in $U\ (\text{respectively, }V)$ is assigned a function from $\mathcal{F}\ (\text{respectively, }\mathcal{G})$. A Holant problem $\Holant(\mathcal{F})$ can seen as the bipartite problem $\Holant({=_2}\myvert \mathcal{F})$.

Let $T$ be a $2\times 2$ matrix and let $\mathcal{F}$ be a function set. Whenever we write $T\mathcal{F}$, the functions in $\mathcal{F}$ are viewed as column vectors and, $T\mathcal{F}=\{T^{\otimes n}f\myvert f\in\mathcal{F}\text{ and }n=\arity(f)\}$. Similarly, $\mathcal{F}T=\{fT^{\otimes n}\myvert f\in\mathcal{F}\text{ and }n=\arity(f)\}$ where the functions in $\mathcal{F}$ are expressed as row vectors.

Let $T$ be a matrix in $GL_2(\mathbb{C})$. We say there is a \emph{holographic reduction} defined by $T$ from $\Holant(\mathcal{F}\myvert\mathcal{G})$ to $\Holant(\mathcal{F}'\myvert\mathcal{G}')$, if $\mathcal{F}T\subseteq \mathcal{F}'$ and $T^{-1}\mathcal{G}\subseteq \mathcal{G}'$. The holographic reduction maps a signature grid $\Omega=(G,\mathcal{F}\myvert\mathcal{G},\pi)$ to $\Omega'=(G,\mathcal{F}'\myvert\mathcal{G}',\pi')$: For each vertex $v$ of $G$, $\pi'$ assigns the function $f_vT$ or $T^{-1}f_v$ to $v$, depending on which part $v$ belongs to.

\begin{thm}[Valiant's Holant Theorem \cite{Valiant2008}]
    Let $T$ be any matrix in $GL_2(\mathbb{C})$. Suppose that the holographic reduction defined by $T$ maps a signature grid $\Omega$ to  $\Omega'$. Then $\Holant_\Omega=\Holant_{\Omega'}$.
\end{thm}

We will use $\redT$ to denote polynomial-time Turing reductions and use $\eqredT$ to denote the equivalence relation under polynomial-time Turing
reductions.

If there is a holographic reduction defined by $T$ from $\Holant(\mathcal{F}\myvert\mathcal{G})$ to $\Holant(\mathcal{F}'\myvert\mathcal{G}')$, then $\Holant(\mathcal{F}\myvert\mathcal{G})\redT \Holant(\mathcal{F}'\myvert\mathcal{G}')$. In particular, if $\mathcal{F}'=\mathcal{F}T$ and $\mathcal{G}'=T^{-1}\mathcal{G}$, then the two problems are equivalent under polynomial-time Turing reductions.

Given a matrix $M$, we use $M^{\mathsf T}$ to denote its transpose. A complex matrix $M$ is orthogonal if $M^{\mathsf T}M=I$, the identity matrix. For any orthogonal matrix $H$, $[1,0,1]H^{\otimes 2}=[1,0,1]$. This gives an important method to normalize a function set.
\begin{thm}
    Let $F$ be a function set and let $H$ be an orthogonal matrix. Then
    \begin{equation*}
        \Holant(H\mathcal{F})\eqredT \Holant(\mathcal{F}).
    \end{equation*}
\end{thm}
\begin{proof}
    Since the equality function $=_2$ is invariant under orthogonal transformation, we have
    \begin{equation*}
        \Holant(H\mathcal{F})\eqredT \Holant({=_2}\myvert H\mathcal{F})\eqredT \Holant({=_2}\myvert \mathcal{F})\eqredT \Holant(\mathcal{F}).
    \end{equation*}
\end{proof}

\subsection{Realizability}

To determine the complexity of $\Holant(\mathcal{F})$ for a given $\mathcal{F}$, a basic technique is realizing certain functions from $\mathcal{F}$. Formally, the notion of realizability is defined by $\mathcal{F}$-gate \cite{CLX2011a}.

Let $\mathcal{F}$ be a set of functions. An $\mathcal{F}$-gate $\Gamma$ is a tuple $(G,\pi)$ where $G=(V,E,D)$ is a graph with regular edges $E$ and some dangling edges $D$ (See \autoref{fig:tetrahedron} for an example). Other than these dangling edges, the gate $\Gamma$ is the same as a signature grid: $\pi$  maps each vertex $v\in V$ to some function $f_v\in\mathcal{F}$ and it incident edges (including the dangling ones) to the input variables of $f_v$. We denote the edges in $E$ by $1,2,...,m$ and the dangling edges in $D$ by $m+1,m+2,...,m+n$. Then we can define a function $f$ for $\Gamma$:
\begin{equation*}
    f(y_1,y_2,...,y_n)=\sum_{x_1,x_2,...,x_m\in\{0,1\}}F(x_1,x_2,...,x_m,y_1,y_2,...,y_n)
\end{equation*}
where $(y_1,y_2,...,y_n)\in\{0,1\}^n$ is an assignment on the dangling edges and $F(\mathbf{x},\mathbf{y})$ denotes the product of evaluations at all vertices of $V$. We say the function $f$ is \emph{realizable} from the function set $\mathcal{F}$. The set $E$ of internal edges is allowed to be empty, in which case $f\in\mathcal{F}$ or it is the tensor product of several functions in $\mathcal{F}$. An $n$-ary function in $\mathcal{F}$ can be seen as a single vertex with $n$ dangling edges (or inputs).

Given a function set $\mathcal{F}$, we define
\begin{equation*}
    \sig(\mathcal{F})=\{g \myvert g\text{ is realizable from }\mathcal{F}\}.
\end{equation*}
When $\mathcal{F}=\{g\}$, we write $\sig(g)$. Note that if two functions $\bar{g}$ and $g$ satisfy $\bar g=cg$ for some constant $c\neq 0$, then $\Holant(\mathcal{F}\cup{\bar{g}})\eqredT\Holant(\mathcal{F}\cup\{g\})$ for any $\mathcal{F}$. Based on this observation, if not confused,  sometimes we also say $\bar{g}$ is realizable from $\mathcal{F}$ if $g\in\sig(\mathcal{F})$.

We list some basic facts on realizable functions.
\begin{lemma}
    If $\mathcal{G}\subseteq\sig(\mathcal{F})$, then $\sig(\mathcal{G})\subseteq \sig(\mathcal{F})$.
\end{lemma}

\begin{lemma}
\label{lemma:realize-red}
    If $f\in\sig(\mathcal{F})$, then $\Holant(\mathcal{F}\cup \{f\})\eqredT \Holant(\mathcal{F})$.
\end{lemma}

Given a function, we can permute its inputs arbitrarily. This will often make a proof clear.
\begin{lemma}
\label{lemma:function-permutation}
    Let $f$ be a realizable function of arity $n$. Then for all permutation $\pi\in S_n$, $f_\pi$ is realizable.
\end{lemma}

If a function is realizable, we also say its signature matrices are realizable.
\begin{lemma}
    If a matrix $M$ is realizable, then $M^{\mathsf T}$ and $MM^{\mathsf T}$ are both realizable.
\end{lemma}

Finally we introduce the notion of \emph{derivative} defined by Cai and Fu \cite{CF2016}. Let $f$ be a function of arity $n$ and let $S$ be a proper subset of $[n]$. We use $\partial_{[x,y]}^S(f)$ to denote the function obtained by connecting one copy of the function $[x,y]$ to each input of $f$ in $S$. In particular, when $S=\{i\}$ for some $i\in [n]$ and $[x,y]=[1,0]\ (\text{or } [0,1])$, we use the simplified notation $f^{x_i=0}\ (\text{or } f^{x_i=1})$.

\subsection{Weighted Counting CSP}
\label{subsec:CSP}
Let $\mathcal{F}$ be a set of complex-valued functions. Then the problem $\CSP(\mathcal{F})$ is defined as follows. An input instance $I$ of the problem consists of
\begin{enumerate}[label=\textbullet]
  \item A finite set of variables $V=\{x_1,...,x_n\}$;
  \item A finite set of constraints $\{C_1,...,C_m\}$: Each has the form $(F_i,\mathbf{x}_i)$ where $F_i\in\mathcal{F}$ and $\mathbf{x}_i$ is a tuple of (not necessarily distinct) variables from $V$.
\end{enumerate}
The instance $I$ \emph{defines} a function of arity $n$:
\begin{equation*}
  F_I(x_1,...,x_n)=\prod_{i=1}^{m} F_i(\mathbf{x}_i).
\end{equation*}
The output is the following sum:
\begin{equation*}
    Z(I)=\sum_{\mathbf{x}\in \{0,1\}^n}F_I(\mathbf{x}).
\end{equation*}

Let $d\geq 1$ be an integer and let $\mathcal{F}$ be a set of complex-valued functions. The problem $\CSP^d(\mathcal{F})$ is the restriction of $\CSP(\mathcal{F})$ where every variable occurs a multiple of $d$ times. This special family of $\CSP$ was first studied in \cite{HL2012}, which played important roles later in proving Holant dichotomies \cite{CGW2013,CFGW2015,CF2016}.

Given a positive integer $k$, we use $=_k$ to denote the following $k$-ary \emph{equality} function:
\begin{equation*}
    f(x_1,x_2,...,x_k)=\begin{cases}
                     1, & \mbox{if } x_1=x_2=\cdots = x_k \\
                     0, & \mbox{otherwise}.
                   \end{cases}
\end{equation*}
The problem $\CSP^d(\mathcal{F})$ is exactly the problem $\Holant(\{=_d,=_{2d},=_{3d},...\}\myvert \mathcal{F})$. Therefore, the Holant framework is expressive enough to simulate $\CSP$.

On the other hand, Holant problems are indeed read-twice \#CSPs. Given a signature grid, we assume that the numbering of its vertices and edges is also given. If these edges are viewed as variables, then the signature grid is a \#CSP instance where every variable appears exactly twice. So we also say that a signature grid defines a function. And the concept of realizability can be defined in the CSP language.

\begin{lemma}
\label{lemma:csp2-Holant}
    $\CSP^2(\mathcal{F})\redT \Holant(\mathcal{F}\cup\{=_4\})$.
\end{lemma}
\begin{proof}
    We have shown that $\CSP^2(\mathcal{F})\eqredT \Holant(\{=_2,=_{4},=_{6},...\}\myvert \mathcal{F})$.

    For $k>2$, supposing that $=_{2(k-1)}$ is realizable, we can obtain $=_{2k}$ by connecting $=_{2(k-1)}$ and $=_4$ via an edge. Thus $\CSP^2(\mathcal{F})\redT \Holant(\mathcal{F}\cup\{=_4\})$.
\end{proof}

Sometimes we need to reduce from $\CSP$ to $\CSP^2$.
\begin{lemma}
\label{lemma:csp-to-csp2}
    Let $\mathcal{F}=\{F_1,...,F_r\}$ be a set of functions. And for each $i\in [r]$, there is a function $f_i$ of arity $m_i$ such that for all $x_1,...,x_{2m_i}\in\{0,1\}$,
    \begin{equation*}
        F_i(x_1,...,x_{2m_i})=f_i(x_1,...,x_{m_i})\prod_{j=1}^{m_i}g(x_j,x_{j+m_i}),
    \end{equation*}
    where $g$ is the binary equality function $=_2$. Then $\CSP(\{f_1,...,f_r\})\redT \CSP^2(\mathcal{F})$.
\end{lemma}
\begin{proof}
    Given an instance $I$ of $\CSP(\{f_1,...,f_r\})$, we construct an instance $I'$ of $\CSP^2(\mathcal{F})$:
    \begin{equation*}
        I'=\{(F_i,\mathbf{x},\mathbf{x})\myvert (f_i,\mathbf{x})\in I\}.
    \end{equation*}
    Then $I$ and $I'$ define the same function. Therefore, $Z(I)=Z(I')$.
\end{proof}

Cai, Lu and Xia \cite{CLX2014} proved a dichotomy for complex-weighted $\CSP$ over the Boolean domain. Before introducing the dichotomy, we need to define two tractable classes of functions.

\begin{mydef}
    The \emph{support} of an $n$-ary function $f$, denoted by $\supp(f)$, is the set $\{\mathbf{x}\in\mathbb{Z}_2^n\myvert f(\mathbf{x})\neq 0\}$.
\end{mydef}

A Boolean relation is \emph{affine} if it is the set of solutions to a system of linear equations over the field $\mathbb{Z}_2$. We say that $f$ has affine support if its support is affine.

\begin{mydef}
    A function $f$ of arity $n$ is \emph{affine} if its support is affine and there is a constant $\lambda\in\mathbb{C}$ such that for all $\mathbf{x}\in\supp(f)$
    \begin{equation*}
        f(\mathbf{x})=\lambda \cdot i^{Q(\mathbf{x})},
    \end{equation*}
    where $i=\sqrt{-1}$ and $Q$ is a homogeneous quadratic polynomial
    \begin{equation*}
        Q(x_1,...,x_n)=\sum_{i=1}^{n}a_ix_i^2+2\sum_{1\leq i<j\leq n} b_{ij}x_ix_j
    \end{equation*}
    with $a_i\in \mathbb{Z}_4$ and $b_{ij}\in\{0,1\}$. We use $\atype$ to denote the set of all affine functions.
\end{mydef}

In particular, if $f\in\atype$ is non-negative and not identically zero, then it has affine support and range $\{0,\lambda\}$ for some $\lambda>0$. Non-negative affine functions are also called \emph{pure affine} in \cite{DGJ2009}.

Let $\neq_2$ denote the binary disequality function $[0,1,0]$.
\begin{mydef}
    A function $f$ is of \emph{product type} if it can be expressed as a product of unary functions, binary functions of the form $=_2$ and $\neq_2$ (on not necessarily disjoint subsets of variables). We use $\ptype$ to denote the set of all functions of product type.
\end{mydef}

\begin{thm}[\cite{CLX2014}]
\label{thm:csp-dichotomy}
    Let $\mathcal{F}$ be a set of complex-valued functions. Then the problem $\CSP(\mathcal{F})$ is computable in polynomial time if $\mathcal{F}\subseteq\atype$ or $\mathcal{F}\subseteq\ptype$. Otherwise $\CSP(\mathcal{F})$ is $\sP$-hard.
\end{thm}

\subsection{Interpolation}
Let $\alpha,\beta\in\mathbb{C}$ be two nonzero complex numbers and $\alpha/\beta$ is not a root of unity. Let $g=[a,b]$ be an unary signature with $ab\neq 0$. Then we can use $H=\begin{bmatrix}\alpha & 0 \\ 0 & \beta\end{bmatrix}$ and $g$ to interpolate all unary signatures in Holant problems:

\begin{lemma}
\label{lemma:interpolate-all-unary}
    For any complex-valued function set $\mathcal{F}$ containing $H$ and $g$,  $\Holant(\mathcal{F}\cup\{[x,y]\})\redT \Holant(\mathcal{F})$ for any $x,y\in\mathbb{C}$.
\end{lemma}
\begin{proof}
    With $H$ and $g$, we can realize a family of unary signatures $\{H^ng\myvert n\in\mathbb{N}\}$. For any $n<m\in\mathbb{N}$, $H^ng$ and $H^mg$ are linearly independent, since
    \begin{align*}
        \det\begin{bmatrix}
              H^ng & H^mg
            \end{bmatrix}=
        \det\begin{bmatrix}
              a\alpha^n & a\alpha^m \\
              b\beta^n & b\beta^m
            \end{bmatrix} = ab\alpha^n\beta^n(\beta^{m-n}-\alpha^{m-n})\neq 0.
    \end{align*}
    Let $\Omega$ be a signature grid of $\Holant(\mathcal{F}\cup\{[x,y]\})$ where $[x,y]$ appears $n$ times. Its Holant can be expressed as a polynomial in $x,y$: $\Holant_\Omega=\sum_{i=0}^{i=n} c_{i}x^iy^{n-i}$. We construct $n+1$ signature grids $\Omega_s\ (s=0,1,...,n)$ of $\Holant(\mathcal{F})$: Each $\Omega_s$ is obtained from replacing all the occurences of $[x,y]$ in $\Omega$ by $g_s=[x_s,y_s]=H^sg$. Then we can interpolate the $c_i$'s by solving the Vandermonde system:
    \begin{align*}
        \Holant_{\Omega_s}=\sum_{i=0}^{n} c_{i}x_s^iy_s^{n-i}, s=0,1,2,...,n.
    \end{align*}
\end{proof}

Using $H$, we can also interpolate all binary signatures of the form $[x,0,y]$, in a similar way as above:
\begin{lemma}
\label{lemma:interpolate-binary-eq}
    For any complex-valued function set $\mathcal{F}$ containing $H$, $\Holant(\mathcal{F}\cup\{[x,0,y]\})\redT \Holant(\mathcal{F})$ for any $x,y\in\mathbb{C}$. In particular, we can interpolate $[1,0]^{\otimes 2}$ and $[0,1]^{\otimes 2}$.
\end{lemma}

In some cases we need to interpolate $=_4$.
\begin{lemma}[\cite{CF2016}]
\label{lemma:interpolate-=_4}
    Suppose $\mathcal{F}$ contains a complex-valued function $f$ of arity 4 with
    \begin{equation*}
        M_{[2]}(f)=\begin{bmatrix}
                       a & 0 & 0 & b \\
                       0 & 0 & 0 & 0 \\
                       0 & 0 & 0 & 0 \\
                       c & 0 & 0 & d
                     \end{bmatrix}.
    \end{equation*}
    where $\begin{bmatrix} a & b \\ c & d \end{bmatrix}$ has full rank. Then $\Holant(\mathcal{F}\cup\{=_4\})\redT \Holant(\mathcal{F})$.
\end{lemma}

\subsection{Holant* Problems}

We introduce several important classes of complex-valued functions over the Boolean domain. $\mathcal{U}$ is the set of all unary functions. And $\mathcal{T}$ is the set of functions of arity at most $2$. Given an element $\mathbf{u}=(u_1,u_2,...,u_n)\in\{0,1\}^n$, we use $\overline{\mathbf{u}}$ to denote the complement $(1-u_1,1-u_2,....,1-u_n)$. $\mathcal{E}$ denotes the function set
\begin{equation*}
    \left\{f\myvert \supp(f)\subseteq\{\mathbf{u},\overline{\mathbf{u}}\}\text{ for some }\mathbf{u}\in\{0,1\}^{\arity(f)}\right\}.
\end{equation*}
And $\mathcal{M}$ denotes the function set in which every function has support consisting of elements of Hamming weight at most one. For example, $[0,1,0,0]\in\mathcal{M}$ but $[1,0,1]\notin \mathcal{M}$.

A function set $\mathcal{F}$ is closed under tensor product if for any $f,g\in\mathcal{F}$ and any permutation $\pi$, $(f\otimes g)_\pi\in\mathcal{F}$. We use $\left<\mathcal{F}\right>$ to denote the least set containing $\mathcal{F}$ that is closed under tenor product. The set $\left<\mathcal{F}\right>$ is called the \emph{tensor closure} of $\mathcal{F}$. Let $\mathscr{T}$ denotes the tensor closure $\left<\mathcal{T}\right>$. Note that the set $\ptype$ of product-type functions, defined in \autoref{subsec:CSP}, is exactly the set $\left<\mathcal{E}\right>$. And for any $2\times 2$ matrix $M$, $M\ptype = \left<M\mathcal{E}\right>$.

$\Holant^*$ problems are the Holant problems where all unary functions are available. Given a function set $\mathcal{F}$, we use $\Holant^*(\mathcal{F})$ to denote the problem $\Holant(\mathcal{F}\cup\mathcal{U})$. Cai, Lu and Xia proved a dichotomy for $\Holant^*$ problems. The modifications of notations have been specified.

\begin{thm}[\cite{CLX2011b}]
\label{thm:Holant*}
    Let $\mathcal{F}$ be any set of complex-valued functions in Boolean variables. The problem $\Holant^*(\mathcal{F})$ is computable in polynomial time if $\mathcal{F}$ satisfies one of the following conditions:
    \begin{enumerate}
      \item $\mathcal{F}\subseteq \ttype$;
      \item $\mathcal{F}\subseteq H \ptype$ for some orthogonal matrix $H$;
      \item $\mathcal{F}\subseteq Z \ptype$ where $Z=\begin{bsmallmatrix} 1 & 1 \\ i & -i\end{bsmallmatrix}$;
      \item $\mathcal{F}\subseteq \left<Z \mathcal{M}\right>$ where $Z=\begin{bsmallmatrix} 1 & 1 \\ i & -i\end{bsmallmatrix}$ or $\begin{bsmallmatrix} 1 & 1 \\ -i & i\end{bsmallmatrix}$.
    \end{enumerate}
    In all other cases, $\Holant^*(\mathcal{F})$ is $\sP$-hard.
\end{thm}

Lemma 6.1 in \cite{CLX2011b} played an important role in the proof of the dichotomy above:
\begin{lemma}[\cite{CLX2011b}]
    Let $\mathcal{F}'$ be any one of $\left<\mathcal{T}\right>$, or $\left<H\mathcal{E}\right>$, or $\left<Z\mathcal{E}\right>$, or $\left<Z\mathcal{M}\right>$. Let $r=3$ if $\mathcal{F'}=\left<\mathcal{T}\right>$, and $r=2$ in the other three cases. Suppose function $F\in\mathcal{F}-\mathcal{F}'$. If $r<\arity(F)$, then we can realize a function $Q$ by connecting $F$ with some unary functions, such that (1) $\Holant^*(\mathcal{F}\cup\{Q\})\eqredT \Holant^*(\mathcal{F})$; (2) $Q\notin\mathcal{F}'$ and (3) $r\leq \arity(Q)<\arity(F)$.
\end{lemma}

We extract the part on non-product-type functions:
\begin{lemma}
\label{lemma:CLX-non-ptype}
    Suppose that $F$ is a function of arity $>2$ and $F\notin\ptype$. Then we can realize a function $Q$ by connecting $F$ with some unary functions, such that $Q\notin\ptype$ and $\arity(Q)=2$.
\end{lemma}

\section{Decomposition}
\label{sec:decomposition}
In Holant problems, sometimes we are able to realize a function $F=f\otimes g$, but do not know how to realize the function $f$ directly, which can be technically beneficial. Fortunately, under certain conditions, if $F$ is realizable, then we may assume that $f$ is freely available.

In this section, we prefer to prove the lemmas in the CSP language. If not specified, the functions we discussed are over a fixed finite domain and take complex values.

Let $m$ be a positive integer. We use $f^{\otimes m}$ to denote the $m$-th tensor power of $f$. $f^{\otimes m}$ can be seen as $m$ copies of $f$: $f^{\otimes m}(\mathbf{x}_1,...,\mathbf{x}_m)=f(\mathbf{x}_1)\cdots f(\mathbf{x}_m)$. Let $I$ be a $\CSP$ instance that contains $m$ constraints: $(f,\mathbf{x}_1),(f,\mathbf{x}_2),....,(f,\mathbf{x}_m)$. We replace these $m$ tuples by one tuple $(f^{\otimes m},\mathbf{x}_1,\mathbf{x}_2,...,\mathbf{x}_m)$ and then obtain a new instance $I'$. It is easy to see that $Z(I)=Z(I')$.

\begin{lemma}
\label{lemma:multiple-to-single}
    For any function set $\mathcal{F}$ and function $f$, $\Holant(\mathcal{F}\cup\{f\})\redT \Holant(\mathcal{F}\cup\{f^{\otimes d}\})$ for all $d\geq 1$.
\end{lemma}
\begin{proof}
    Impose induction on $d$. Let $n$ denote the arity of $f$.

    The base case, $d=1$, is trivial. Now suppose that the conclusion holds for all $d<k\ (k\geq 2)$. In the problem $\Holant(\mathcal{F}\cup\{f^{\otimes k}\})$, we may assume that the functions $f^{\otimes (mk)}$ are freely available for integers $m>0$. There are two cases to consider:
    \begin{enumerate}[label=(\arabic*),leftmargin=*]
      \item There exists an instance $I$ of $\Holant(\mathcal{F}\cup\{f\})$ such that $Z(I)\neq 0$ and $f$ appears $p$ times where $p=qk+r\ (q\geq 0, 0<r<k)$. Let $C_1,...,C_p$ be the $p$ constraints that have the form $(f,\mathbf{x}_i)$. We replace the first $qk$ constraints by one tuple $C_1'=(f^{\otimes (qk)},\mathbf{x}_1,...,\mathbf{x}_{qk})$, and the last $r$ constraints by one tuple $C_2'=(f^{\otimes k}, \mathbf{x}_{qk+1},...,\mathbf{x}_{p},\mathbf{y})$ where $\mathbf{y}$ denotes a list of new distinct variables, of length $(k-r)n$. After the substitution, we get a function $F(\mathbf{x},\mathbf{y})$ where $\mathbf{x}$ denotes the variables of the original instance $I$. Every variable in $\mathbf{x}$ occurs twice, so by summing on them we can realize the following function:
          \begin{equation*}
            \sum_{\mathbf{x}}F(\mathbf{x},\mathbf{y})=\sum_{\mathbf{x}}F_I(\mathbf{x})f^{\otimes (k-r)}(\mathbf{y})=Z(I)f^{\otimes (k-r)}(\mathbf{y}).
          \end{equation*}
          Because $Z(I)\neq 0$, we have $\Holant(\mathcal{F}\cup\{f^{\otimes (k-r)}\})\redT \Holant(\mathcal{F}\cup\{f^{\otimes k}\})$. And by the induction hypothesis, $\Holant(\mathcal{F}\cup\{f\})\redT \Holant(\mathcal{F}\cup\{f^{\otimes (k-r)}\})$. Therefore, the conclusion holds.
      \item For all $I$ with $Z(I)\neq 0$, $f$ appears a multiple of $k$ times. Given an instance $I$ of $\Holant(\mathcal{F}\cup\{f\})$, we show how to compute $Z(I)$ with the help of the oracle for $\Holant(\mathcal{F}\cup\{f^{\otimes k}\})$. First we check whether the number $p$ of constraints containing $f$ is a multiple of $k$. If not, we simply output 0. Otherwise we replace all such constraints by one tuple $(f^{\otimes p},\mathbf{x})$ as in case (1), and then obtain an instance $I'$ of $\Holant(\mathcal{F}\cup\{f^{\otimes k}\})$. Clearly $Z(I)=Z(I')$, and we can compute $Z(I')$ by accessing the oracle.
    \end{enumerate}
    In either case, there exists a polynomial-time Turing reduction. This completes the induction.
\end{proof}

Now we are ready to prove a more general lemma.
\begin{lemma}
\label{lemma:decomposition-lemma}
    Let $\mathcal{F}$ be a set of functions, and $f,g$ be two functions. Suppose that there exists an instance $I$ of $\Holant(\mathcal{F}\cup\{f,g\})$ such that $Z(I)\neq 0$, and the number of occurrences of $g$ in $I$ is greater than that of $f$. Then
    \begin{equation*}
        \Holant(\mathcal{F}\cup\{f,f\otimes g\})\redT \Holant(\mathcal{F}\cup\{f\otimes g\}).
    \end{equation*}
\end{lemma}
\begin{proof}
    For each pair of constraints $(f,\mathbf{x}_i)$ and $(g,\mathbf{x}_j)$ in the instance $I$, we replace them by one tuple $(f\otimes g,\mathbf{x}_i,\mathbf{x}_j)$. Let $(g,\mathbf{z}_1),...,(g,\mathbf{z}_r)$ be the unpaired constraints where $r>0$. Replace each tuple $(g,\mathbf{z}_i)$ by a new tuple $(f\otimes g,\mathbf{y}_i,\mathbf{z}_i)$ where $\mathbf{y}_i$ is a set of distinct variables that do not appear in $I$. After the substitution, we get a function $F(\mathbf{x},\mathbf{y})$ where $\mathbf{x}$ denotes the variables of $I$ and $\mathbf{y}=(\mathbf{y}_1,...,\mathbf{y}_r)$. As in the proof of \autoref{lemma:multiple-to-single}, we can realize the function $Z(I)f^{\otimes r}$ by summing on $\mathbf{x}$. Then the conclusion follows from \autoref{lemma:multiple-to-single} and \autoref{lemma:realize-red}:
    \begin{equation*}
        \Holant(\mathcal{F}\cup\{f,f\otimes g\})\redT \Holant(\mathcal{F}\cup\{f^{\otimes r},f\otimes g\})\redT \Holant(\mathcal{F}\cup\{f\otimes g\}).
    \end{equation*}
\end{proof}

Note that our proofs for the two lemmas above only show the \emph{existence} of polynomial-time Turing reductions, but do not produce such reductions \emph{constructively} for given function sets.

The condition of \autoref{lemma:decomposition-lemma} seems somewhat complicated. We can make it more stringent and hence, easier to apply. A function $f$ is \emph{vanishing} \cite{CGW2013}, if $Z(I)=0$ for every instance $I$ of the problem $\Holant(f)$.
\begin{corollary}
\label{coro:decomposition-non-vanishing}
    Let $\mathcal{F}$ be a set of functions, and $f,g$ be two functions. If $g$ is not vanishing, then
    \begin{equation*}
        \Holant(\mathcal{F}\cup\{f,f\otimes g\})\redT \Holant(\mathcal{F}\cup\{f\otimes g\}).
    \end{equation*}
\end{corollary}

One more lemma concludes this section.

\begin{lemma}
\label{lemma:Holant-multiple}
    For positive integers $d_1,d_2,...,d_n$,
    \begin{equation*}
        \Holant(\mathcal{F}\cup\{f_1,f_2,...,f_n\})\redT \Holant(\mathcal{F}\cup\{f_1^{\otimes d_1},f_2^{\otimes d_2},...,f_n^{\otimes d_n}\}).
    \end{equation*}
\end{lemma}
\begin{proof}
    From \autoref{lemma:multiple-to-single} we see the reduction chain:
    \begin{align*}
      \Holant(\mathcal{F}\cup\{f_1,f_2,...,f_n\}) &\redT  \Holant(\mathcal{F}\cup\{f_1^{\otimes d_1},f_2,...,f_n\}) \\
      &\redT  \Holant(\mathcal{F}\cup\{f_1^{\otimes d_1},f_2^{\otimes d_2},...,f_n\}) \\
      & \cdots \\
      &\redT \Holant(\mathcal{F}\cup\{f_1^{\otimes d_1},f_2^{\otimes d_2},...,f_n^{\otimes d_n}\}).
    \end{align*}
\end{proof}

\section{When A Non-trivial Equality Function Appears}
\label{sec:non-unitary}

Let $\Holant^c(\mathcal{F})$ denote the problem $\Holant(\mathcal{F}\cup\{[1,0],[0,1]\})$. In this section, we will prove the following theorem:

\begin{thm}
\label{thm:Holantc-nontrivial-eq}
    Let $\lambda$ be any nonzero complex number that is not a root of unity. For any set $\mathcal{F}$ of complex-valued functions, $\Holant^c(\mathcal{F}\cup\{[1,0,\lambda]\})$ is computable in polynomial time if $\mathcal{F}\subseteq \ttype$ or $\mathcal{F}\subseteq \ptype$. Otherwise the problem is $\sP$-hard.
\end{thm}

The conclusion still holds if we remove the unary functions $[1,0]$ and $[0,1]$:
\begin{thm}
\label{thm:Holant-nontrivial-eq}
    Let $\lambda$ be any nonzero complex number that is not a root of unity. For any set $\mathcal{F}$ of complex-valued functions, $\Holant(\mathcal{F}\cup\{[1,0,\lambda]\})$ is computable in polynomial time if $\mathcal{F}\subseteq \ttype$ or $\mathcal{F}\subseteq \ptype$. Otherwise the problem is $\sP$-hard.
\end{thm}
\begin{proof}
    By \autoref{lemma:interpolate-binary-eq}, we can interpolate $[1,0]^{\otimes 2}$ and $[0,1]^{\otimes 2}$ using $[1,0,\lambda]$. Then by \autoref{lemma:Holant-multiple}, $\Holant^c(\mathcal{F}\cup\{[1,0,\lambda]\})\redT \Holant(\mathcal{F}\cup\{[1,0,\lambda]\})$.
\end{proof}

Throughout this section, $\lambda$ is a nonzero complex number that is not a root of unity, and all the functions we discussed are complex-valued. We use $\mathcal{F}$ to denote a function set.

\subsection{Parity Condition}

A function has \emph{even} (\emph{odd}) support if the support does not contain any elements of odd (even) Hamming weight.

\begin{mydef}[Parity Condition]
    A function satisfies the \emph{Parity} condition if it has even or odd support. A function set $\mathcal{F}$ satisfies the Parity condition if every function in $\mathcal{F}$ does.
\end{mydef}

\begin{lemma}
\label{lemma:odd-irreducible-ptype}
    Let $f$ be a function of odd arity $n$, with support $\{\mathbf{u},\overline{\mathbf{u}}\}$ for some $\mathbf{u}\in\{0,1\}^n$. Then $[x,y]\in\sig(f)$ for some $xy\neq 0$.
\end{lemma}
\begin{proof}
    To simplify the notation, we assume that $\mathbf{u}=0^s1^t$ for some $s,t\geq 0$. Since $s+t=n$ is odd, $s$ and $t$ have opposite parity. Suppose that $s$ is odd (the other case is similar). Let $m=(n-1)/2$, then the function
    \begin{equation*}
        h(x)=\sum_{x_1,x_2...,x_m\in\{0,1\}} f(x,x_1,x_1,x_2,x_2,...,x_m,x_m)
    \end{equation*}
    has the signature $[f(\mathbf{u}),f(\overline{\mathbf{u}})]$.
\end{proof}

\begin{lemma}
\label{lemma:parity-unsat}
    If $\mathcal{F}$ does not satisfy the Parity condition, then $[x,y]\in\sig(\mathcal{F}\cup\{[1,0],[0,1]\})$ for some $xy\neq 0$.
\end{lemma}
\begin{proof}
    Let $f\in\mathcal{F}$ be a function of arity $n$ that does not satisfy the Parity condition. Given $\mathbf{u}\in\{0,1\}^n$, we use $w(\mathbf{u})$ to denote its Hamming weight. Then there are some $\mathbf{a}=a_1\cdots a_n,\mathbf{b}=b_1\cdots b_n\in \supp(f)$ whose Hamming weights are of opposite parity, and
    \begin{equation*}
        m=w(\mathbf{a} \oplus \mathbf{b})=\min
        \{w(\mathbf{c}\oplus \mathbf{d})\myvert \mathbf{c},\mathbf{d}\in \supp(f), w(\mathbf{c})\text{ and } w(\mathbf{d}) \text{ have opposite parity}\}.
    \end{equation*}
    We define two sets: $S_c=\{k\in [n] \myvert a_k=b_k=c\}$ for $c\in\{0,1\}$. Then the function
    \begin{equation*}
        g=\partial_{[1,0]}^{S_0}(\partial_{[0,1]}^{S_1}(f))
    \end{equation*}
    has \emph{odd} arity $m$ and its support is $\{\mathbf{u},\overline{\mathbf{u}}\}$ for some $\mathbf{u}\in\{0,1\}^m$. By \autoref{lemma:odd-irreducible-ptype}, we can realize an unary function $[x,y]$ with $xy\neq 0$.
\end{proof}

Let $f$ be a signature of arity $n$ and let $T$ be a $2\times 2$ matrix. Recall that we use $Tf$ to denote the signature $T^{\otimes n}f$ where $f$ is viewed as a column vector. The follow lemma is a corollary of \autoref{thm:Holant*}.
\begin{lemma}
\label{lemma:Holant*-nontrivial-eq}
    $\Holant^*(\mathcal{F}\cup\{[1,0,\lambda]\})$ is $\sP$-hard or $\mathcal{F}\subseteq \ttype$ or $\mathcal{F}\subseteq \ptype$.
\end{lemma}
\begin{proof}
    Suppose that $\Holant^*(\mathcal{F}\cup\{[1,0,\lambda]\})$ is not $\sP$-hard. Then by \autoref{thm:Holant*}, at least one of the following conditions holds:
    \begin{enumerate}[label=(\arabic*),labelindent=*,leftmargin=*]
      \item $\mathcal{F}\subseteq \ttype$.
      \item $\mathcal{F}\cup\{[1,0,\lambda]\}\subseteq H\ptype$ for some orthogonal matrix $H$. Since $\begin{bsmallmatrix} 1 & 0 \\ 0 & -1\end{bsmallmatrix}\ptype\subseteq\ptype$, we may suppose that $H=\begin{bsmallmatrix} x & y \\ y & -x \end{bsmallmatrix}$ where $x^2+y^2=1$. Thus $H^{-1}[1,0,\lambda]=[x^2+\lambda y^2, (1-\lambda)xy, y^2+\lambda x^2]\in\ptype$. So we have $x^2+\lambda y^2=y^2+\lambda x^2=0$ or $(1-\lambda)xy=0$. If $x^2+\lambda y^2=y^2+\lambda x^2=0$, then $x^2+\lambda y^2 + y^2+\lambda x^2=1+\lambda=0$ hence $\lambda=-1$, contradicting the fact that $\lambda$ is not a root of unity. So $(1-\lambda)xy=0$, which implies that $x=0$ or $y=0$. Thus $H\ptype\subseteq\ptype$. And $\mathcal{F}\subseteq \ptype$.
      \item $[1,0,\lambda]\in Z\ptype$. But $Z^{-1}[1,0,\lambda]=\frac{1}{4}[1-\lambda,1+\lambda,1-\lambda]\notin \ptype$.
      \item $[1,0,\lambda]\in Z\mathcal{M}$. But $Z^{-1}[1,0,\lambda]=\frac{1}{4}[1-\lambda,1+\lambda,1-\lambda]\notin \mathcal{M}$.
    \end{enumerate}
\end{proof}

\begin{thm}
\label{thm:Holantc-nontrivial-eq-parity}
    Suppose that $\mathcal{F}$ does not satisfy the Parity condition. Then $\Holant^c(\mathcal{F}\cup\{[1,0,\lambda]\})$ is $\sP$-hard or $\mathcal{F}\subseteq\ttype$ or $\mathcal{F}\subseteq\ptype$.
\end{thm}
\begin{proof}
    By \autoref{lemma:parity-unsat}, $[x,y]\in\sig(\mathcal{F}\cup\{[1,0],[0,1]\})$ for some $xy\neq 0$. Thus we can interpolate all unary functions by \autoref{lemma:interpolate-all-unary}. Then the conclusion follows from \autoref{lemma:Holant*-nontrivial-eq}.
\end{proof}

\subsection{Non-product-type Functions}

Let $f$ be a function of arity $n>0$. $f$ can be seen as a gate with $n$ inputs. First we define two binary relations (depending on $f$), $E_f$ and $N_f$, on the set $[n]$: for $i,j\in [n]$,
\begin{enumerate}[label=\textbullet]
  \item $(i,j)\in E_f$ if for all $x_1,...,x_n\in\{0,1\}$, $f(x_1,...,x_n)=0$ when $x_i\neq x_j$;
  \item $(i,j)\in N_f$ if for all $x_1,...,x_n\in\{0,1\}$, $f(x_1,...,x_n)=0$ when $x_i=x_j$.
\end{enumerate}
And we denote the relation $E_f\cup N_f$ by $\sim_f$. It is easy to verify that $\sim_f$ is an equivalence relation, so it determines a partition of $[n]$. We denote the partition by $\inp{f}$, called the \emph{input partition} of the function $f$.

The following lemma tells us that we can reduce the arity of a non-product-type function by pinning. The idea of the proof is similar to that of Lemma 5.8 in \cite{CLX2014}. But here the unary function $[1,1]$ is not freely available, so the realizable functions are slightly different.
\begin{lemma}
\label{lemma:non-ptype}
    Let $F\notin \ptype$ be a function with affine support. Then there is a function $g\in\sig(\{F,[1,0],[0,1]\})$ which has one of the following forms:
    \begin{enumerate}[label=(\arabic*),labelindent=*,leftmargin=*]
      \item $g(x_1,...,x_m)=h(x_1,x_2,x_3)\prod_{i=1}^{m-3}h_i(x_i,x_{i+3})$ where $m\in\{3,6\}$, $h_i\in\{=_2,\neq_2\}$ and the support of $h$ is determined by an equation over $\mathbb{Z}_2$: $x_1\oplus x_2\oplus x_3=c$ for some $c\in\{0,1\}$. Therefore, $g,g^2\notin\ptype$. Note that if $m=3$, then the part of $h_i$ disappears and $g$ is simply the function $h$.
      \item $g=\begin{bmatrix}a & b \\ c & d\end{bmatrix}$ or $\begin{bsmallmatrix} a & 0 & 0 & b \\ 0 & 0 & 0 & 0 \\ 0 & 0 & 0 & 0 \\ c & 0 & 0 & d \end{bsmallmatrix}$ where $abcd\neq 0$ and $(a,b,c,d)\notin\ptype$.
    \end{enumerate}
\end{lemma}
\begin{proof}
    We may suppose that $F\neq F'\otimes \Delta$ for any functions $F'$ and $\Delta\in\{[1,0],[0,1]\}$. Otherwise we can obtain $F'$ by pinning. $F'\notin\ptype$ and it has affine support, thus we can consider the function $F'$ instead. Since $F\notin\ptype$, $\inp{F}=\{I_1,...,I_k\}$ for some $k>1$.

    Let $s$ denote the dimension of the support of $F$. Then $s>1$ otherwise $F\in\ptype$. We use $\{y_1,...,y_n\}$ to denote the input variables of $F$ where $n=\arity(F)$. $y_1,...,y_n$ satisfy a system of linear equations over $\mathbb{Z}_2$ (the solutions constitute the support of $F$). Since the inputs of $F$ can be permuted arbitrarily, we may assume that $\{y_1,...,y_s\}$ is a set of free variables and $i\in I_i$ for all $i\in [s]$. There are two cases:
    \begin{enumerate}[label=(\arabic*),align=left,leftmargin=0pt,listparindent=\parindent,labelwidth=0pt,itemindent=!]
      \item $s<k$. Let $r$ be an index in $I_k$. Then on the support of $F$, it holds that $y_r=\sum_{i=1}^{s} a_i y_i+b\ (\text{mod }2)$ where $a_i,b\in\{0,1\}$. Moreover, at least two of ${a_i}$'s are nonzero because $k>s$. Suppose that $a_1a_2\neq 0$. We define two sets:
          \begin{align*}
            S_0 &= \{i\in [n]\myvert y_i=0\text{ if } y_3=y_4=\cdots =y_s=0\}, \\
            S_1 &= \{i\in [n]\myvert y_i=1\text{ if } y_3=y_4=\cdots =y_s=0\}.
          \end{align*}
          Then $g=\partial_{[1,0]}^{S_0}(\partial_{[0,1]}^{S_1}(F))$ has affine support of dimension 2 and $\inp{g}=\{J_1,J_2,J_3\}$. If  $|J_i|\geq 3$ for some $i$, then there must be two inputs of $g$, say the $p$th and the $q$th, such that $p,q\in J_i$ and $(p,q)\in E_g$. We connect them via an edge and then obtain a new function $g'$ of lower arity. By the definition of $\inp{g}$, $g'$ also has affine support of dimension $2$ and $|\inp{g'}|=3$. Therefore, we may suppose that $1\leq |J_i|\leq 2$ for all $i\in\{1,2,3\}$. Again, since the inputs of $g$ can be permuted, we further suppose that $i\in J_i$ for all $i\in\{1,2,3\}$. Let $\{x_1,...,x_m\}$ denote the input variables of $g$ where $m=\arity(g)$, then on the support of $g$, $x_1,x_2,x_3$ satisfy an equation $x_1\oplus x_2\oplus x_3=c$ for some $c\in\{0,1\}$.

          If $|J_i|=1$ for all $i$ or $|J_i|=2$ for all $i$, then we are done. Otherwise $g$ does not satisfy the Parity condition. By \autoref{lemma:parity-unsat}, we can realize an unary function $[x,y]$ with $xy\neq 0$. Let $S$ denote the set $\{i\in [m]\myvert i>3\}$. Then $\partial_{[x,y]}^S(g)$ is a ternary function whose support is determined by the equation $x_1\oplus x_2\oplus x_3=c$.
      \item $s=k$. Since $\{y_1,...,y_k\}$ is a set of free variables, there is a function $f$ such that for all $y_1,...,y_n\in\{0,1\}$,
          \begin{equation*}
            F(y_1,...,y_n)=\chi_F(y_1,...,y_n)\cdot f(y_1,...,y_k),
          \end{equation*}
          where $\supp(f)=\{0,1\}^k$ and $\chi_F$ denotes the Boolean function defined by the support of $F$. $f$ can not be degenerate, otherwise $F\in\ptype$. Therefore, there is some $i\in [k]$ and $\mathbf{u},\mathbf{v}\in\{0,1\}^{k-1}$, such that
          \begin{equation*}
                \frac{f^{y_i=1}(\mathbf{u})}{f^{y_i=0}(\mathbf{u})}\neq \frac{f^{y_i=1}(\mathbf{v})}{f^{y_i=0}(\mathbf{v})}.
          \end{equation*}
          Because $\supp(f)=\{0,1\}^k$, we may assume that $\mathbf{u}$ and $\mathbf{v}$ are adjacent. That is, the bitwise XOR $\mathbf{u}\oplus\mathbf{v}$ is of Hamming weight 1. Without loss of generality, we further suppose that $i=1$, $\mathbf{u}=0u_3\cdots u_k$ and $\mathbf{v}=1u_3\cdots u_k$. Then the function $h=f^{y_3=u_3,...,y_k=u_k}$ has the signature $\begin{bmatrix}a & b \\ c & d\end{bmatrix}$ where $abcd\neq 0$ and $ad\neq bc$. Thus $h\notin \ptype$.

          As in case (1), we can connect two inputs of $F$ that must take the same value. So we assume that for all $j\in [k]$, $1\leq |I_j|\leq 2$. Further, there are three subcases to consider:
          \begin{enumerate}[label=(\roman*),align=left,leftmargin=0pt,listparindent=\parindent,labelwidth=0pt,itemindent=!]
            \item For all $j\in [k]$, $|I_j|=1$. Then $F=f$. And we have shown that the function $h=(a,b,c,d)$ is realizable, with $abcd\neq 0$ and $h\notin \ptype$.
            \item There are two indices $p,q\in [k]$ such that $|I_p|=1$ and $|I_q|=2$. In this case, $F$ does not satisfy the Parity condition. Applying \autoref{lemma:parity-unsat}, we can realize an unary function $[x,y]$ with $xy\neq 0$. Let $S$ denote the set $\{i\in [n]\myvert i>k\}$. Then we only need to consider the function $F'=\partial^S_{[x,y]}(F)$, back to the case (i).
            \item For all $j\in [k]$, $|I_j|=2$. For each $3\leq j\leq k$ and each $i\in I_j$, if $(i,j)\in E_F$, we pin the $i$th input of $F$ to $u_j$, otherwise we pin the input to $\overline{u_j}$. This produces a function $H(x_1,x_2,x_3,x_4)=h(x_1,x_3)f_1(x_1,x_2)f_2(x_3,x_4)$, where $f_1,f_2\in\{=_2,\neq_2\}$. If both $f_1$ and $f_2$ are the equality function, then we are done. Otherwise, by pinning we can realize a general disequality function $(0,x,y,0)$ with $xy\neq 0$. With this function, we are able to flip $x_2$ or $x_4$, and realize a function $H'(x_1,x_2,x_3,x_4)=h(x_1,x_3)f_1'(x_1,x_2)f_2'(x_3,x_4)$ where $f_1',f_2'\in\{=_2,[x,0,y]\}$.
          \end{enumerate}
    \end{enumerate}
\end{proof}

\subsection{Hardness Proof}

This subsection is devoted to the hardness part of \autoref{thm:Holantc-nontrivial-eq}. Before this, we need to make some preparations.

The complete dichotomy for sets of symmetric signatures \cite{CGW2013} implies the following lemma. For completeness, we give a proof.
\begin{lemma}
\label{lemma:{[1,0,1,0],[1,0,lambda]}}
    $\Holant(\{[1,0,1,0],[1,0,\lambda]\})$ is $\sP$-hard.
\end{lemma}
\begin{proof}
    Note that $[1,0,1,0]=\frac{1}{2}([1,1]^{\otimes 3}+[1,-1]^{\otimes 3})$, thus by performing the orthogonal transformation $H=\frac{1}{\sqrt{2}}\begin{bsmallmatrix} 1 & 1 \\ 1 & -1 \end{bsmallmatrix}$, we have
    \begin{equation*}
        \Holant(\{=_3,[1+\lambda,1-\lambda,1+\lambda]\})\redT \Holant(\{[1,0,1,0],[1,0,\lambda]\}).
    \end{equation*}
    For any integer $k>3$, the equality function $=_k$ can be realized using $k-2$ ternary equality functions. So $\CSP([1+\lambda,1-\lambda,1+\lambda])\redT \Holant(\{=_3,[1+\lambda,1-\lambda,1+\lambda]\})$. Since $\lambda$ is nonzero and not a root of unity, $[1+\lambda,1-\lambda,1+\lambda]\notin \atype\cup\ptype$. This implies that $\CSP([1+\lambda,1-\lambda,1+\lambda])$ is $\sP$-hard by \autoref{thm:csp-dichotomy}. Therefore, so is the problem $\Holant(\{[1,0,1,0],[1,0,\lambda]\})$.
\end{proof}

\begin{mydef}
    Let $f,g$ be two functions of arity $n$. $f$ and $g$ are \emph{equivalent}, denoted by $f\sim g$, if
    \begin{enumerate}[leftmargin=*,labelindent=\parindent]
      \item $g=f_\pi$ for some $\pi\in S_n$, or
      \item $g(x_1,x_2,...,x_n)=f(\overline{x_1},\overline{x_2},...,\overline{x_n})$ for all $x_1,x_2,...,x_n\in\{0,1\}$.
    \end{enumerate}
\end{mydef}

\begin{lemma}
    If $f\sim g$, then $\Holant(f)\eqredT \Holant(g)$.
\end{lemma}

The lemma above implies that, if $f\sim g$, then $\Holant(f)$ is $\sP$-hard if and only if $\Holant(g)$ is $\sP$-hard.

Proceed to prove the hardness. First we consider irreducible ternary functions.
\begin{lemma}
\label{lemma:Holant-nontrivial-eq-ternary}
    Let $f\notin\ptype$ be any irreducible function of arity $3$. Then $\Holant^c(\{f,[1,0,\lambda]\})$ is $\sP$-hard.
\end{lemma}
\begin{proof}
    Suppose that $f$ satisfies the Parity condition; otherwise we are done by \autoref{thm:Holantc-nontrivial-eq-parity}. Up to the relation $\sim$, there are two cases:
    \begin{enumerate}[label=(\arabic*),leftmargin=*]
      \item $f=(x,0,0,y,0,z,w,0)$ with $xzw\neq 0$. Let $T=\begin{bsmallmatrix} 1 & 0 \\ 0 & \alpha\end{bsmallmatrix}$ where $\alpha$ will be determined later. Then $f_\alpha=T^{\otimes 3}f=(x,0,0,\alpha^2y,0,\alpha^2z,\alpha^2w,0)$. Consider the triangle in \autoref{fig:triangle}. It has the symmetric signature $g=[x^3+\alpha^6y^3,0,(x+\alpha^2y)\alpha^4zw,0]=[a,0,b,0]$. We can choose a suitable $\alpha\neq 0$ such that $ab\neq 0$. Now taking $M=\begin{bsmallmatrix} 1 & 0 \\ 0 & \beta\end{bsmallmatrix}$ where $\beta^2=a/b$, we obtain $h=M^{\otimes 3}g=[a,0,b\beta^2,0]=a[1,0,1,0]$. With $[1,0,\lambda]$ at hand, we can interpolate both $T$ and $M$. Thus $\Holant(\{[1,0,1,0],[1,0,\lambda]\})\redT \Holant^c(\{f,[1,0,\lambda]\})$. Since the former problem is \#P-hard by \autoref{lemma:{[1,0,1,0],[1,0,lambda]}}, so is the latter.
          \begin{figure}[h]
            \centering
            \includegraphics[scale=0.6]{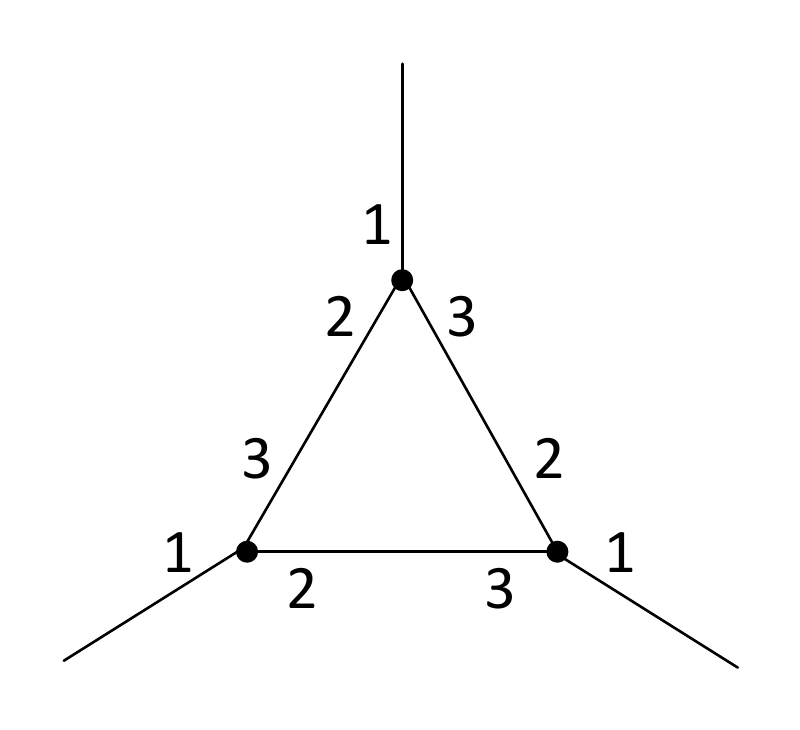}
            \caption{Three vertices are assigned the function $f_\alpha$.}\label{fig:triangle}
          \end{figure}
      \item $f=(0,0,0,x,0,y,z,0)$ with $xyz\neq 0$. Again, using the triangle structure, we obtain $g=[x^3,0,xyz,0]=[a,0,b,0]$. Then the proof is similar to that of case (1).
    \end{enumerate}
\end{proof}

Now we come to the main part of the hardness.

\begin{lemma}
\label{lemma:Holantc-nontrivial-eq-single}
    Let $f\notin \ptype$ be a function of arity $n\geq 3$. If $f$ satisfies the Parity condition, then $\Holant^c(\{f,[1,0,\lambda]\})$ is $\sP$-hard.
\end{lemma}
\begin{proof}
    First consider the case $f^2\notin \ptype$. By \autoref{lemma:CLX-non-ptype}, and to simplify the notation, we assume that there are $n-2$ unary functions $u_i=[x_i,y_i]\ (i\in[n-2])$, such that $g=\partial_{u_1}^{\{1\}}\partial_{u_2}^{\{2\}}\cdots \partial_{u_{n-2}}^{\{n-2\}}(f^2)\notin \ptype$. With $[1,0,\lambda]$ at hand, we interpolate $v_i=[x_i,0,y_i]$ for each $i\in[n-2]$. Now take two copies of $f$. For each $i\in [n-2]$, we connect $i$th inputs of the two copies via $v_i$. This realizes a function $G$ of arity $4$, such that
    \begin{equation*}
        G(x,x,y,y)=g(x,y) \text{ for all }x,y\in\{0,1\}.
    \end{equation*}
    Because $f$ and $v_i\ (i\in [n-2])$ satisfy the Parity condition, so does $G$. Note that $g$ is binary and $g\notin\ptype$, so $g$ is non-degenerate, and at least three of $\{0000,0011,1100,1111\}$ belong to $\supp(G)$. Therefore, $G$ has even support, whose $4\times 4$ signature matrix is
    \begin{equation*}
        M_G=\begin{bmatrix}
                       a & 0 & 0 & b \\
                       0 & x & y & 0 \\
                       0 & z & w & 0 \\
                       c & 0 & 0 & d
                     \end{bmatrix}
    \end{equation*}
    where $g=\begin{bmatrix}a & b \\ c & d\end{bmatrix}$.  If one of $x,y,z$ or $w$ is nonzero, then at least one of the ternary functions in $\{G^{x_i=j}\myvert i=1,2,3,4\text{ and }j=0,1\}$, say $h$, is irreducible and not of product type. By \autoref{lemma:Holant-nontrivial-eq-ternary}, $\Holant^c(\{h,[1,0,\lambda]\})$ is \#P-hard. Now suppose that $x=y=z=w=0$, then we can interpolate $=_4$ by \autoref{lemma:interpolate-=_4}. Hence we have the following reductions:
    \begin{align*}
        \CSP^2(\{G,[1,0,\lambda]\}) &\redT \Holant(\{G,[1,0,\lambda],=_4\})\\
         &\redT \Holant^c(\{f,[1,0,\lambda]\}).
    \end{align*}
    The first reduction follows from \autoref{lemma:csp2-Holant}. And by \autoref{lemma:csp-to-csp2}, we have
    \begin{equation*}
        \CSP(\{g,[1,\lambda]\})\redT \CSP^2(\{G,[1,0,\lambda]\}).
    \end{equation*}
    Since $g\notin\ptype$ and $[1,\lambda]\notin\atype$, $\CSP(\{g,[1,\lambda]\})$ is $\sP$-hard. Therefore, so is $\Holant^c(\{f,[1,0,\lambda]\})$.

    Now we suppose that $f^2\in\ptype$ and hence the support of $f$ is affine. Then by \autoref{lemma:non-ptype}, we can realize a non-product-type function $g$. Since the set $\{f,[1,0],[0,1]\}$ satisfies the Parity condition, so $g$ can not be of the form $\begin{bmatrix}a & b \\ c & d\end{bmatrix}$ in \autoref{lemma:non-ptype}. There are two cases:
    \begin{enumerate}[label=(\arabic*),leftmargin=*]
      \item $g^2\notin\ptype$. We have shown that $\Holant^c(\{g,[1,0,\lambda]\})$ is $\sP$-hard.
      \item $g=\begin{bsmallmatrix} a & 0 & 0 & b \\ 0 & 0 & 0 & 0 \\ 0 & 0 & 0 & 0 \\ c & 0 & 0 & d \end{bsmallmatrix}$ where $abcd\neq 0$ and $(a,b,c,d)\notin\ptype$. Then we can use $g$ to interpolate $=_4$. Again, we have the reduction
          \begin{equation*}
            \CSP(\{(a,b,c,d),[1,\lambda]\}) \redT \Holant^c(\{f,[1,0,\lambda]\}).
          \end{equation*}
          Therefore, the problem $\Holant^c(\{f,[1,0,\lambda]\})$ is $\sP$-hard.
    \end{enumerate}
\end{proof}

\autoref{lemma:Holantc-nontrivial-eq-single} requires that $f$ satisfy the Parity condition. Only the conditions $f\notin\ptype$ and $\arity(f)\geq 3$ are not sufficient for hardness; it is possible that $f\in\ttype$. The following lemma explains why the Parity condition (assuming $f\notin\ptype$) excludes the case $f\in\ttype$.
\begin{lemma}
\label{lemma:parity-T-P}
    Let $f$ be a function satisfying the Parity condition. If $f\in \ttype$, then $f\in \ptype$.
\end{lemma}
\begin{proof}
    Suppose that $f\in\ttype$ and it is not identically zero. Then $f=f_1\otimes\cdots\otimes f_k$ where $f_i$ is of arity $\leq 2$ for all $i\in [k]$. Every $f_i$ is not identically zero, thus $f_i\in\sig(\{f,[1,0],[0,1]\})$ for all $i\in [k]$ and hence they all satisfy the Parity condition. This means that $f_i\in \ptype$ all $i\in [k]$. Note that $\ptype$ is closed under tensor product, so $f\in\ptype$.
\end{proof}

To conclude this section, we prove \autoref{thm:Holantc-nontrivial-eq}. In fact, we have done most of the work in \autoref{lemma:Holantc-nontrivial-eq-single}.

\begin{proof}[Proof of \autoref{thm:Holantc-nontrivial-eq}]
    Given a function set $\mathcal{F}$, we suppose that $\mathcal{F}$ satisfies the Parity condition, otherwise we are done by \autoref{thm:Holantc-nontrivial-eq-parity}. Now suppose that $\mathcal{F}\not\subseteq\ptype$. Let $f$ be any function in $\mathcal{F}\backslash\ptype$. By \autoref{lemma:parity-T-P}, $f\notin \ttype$ since $f$ satisfies the Parity condition. So $f$ is of arity $\geq 3$. \autoref{lemma:Holantc-nontrivial-eq-single} shows that $\Holant^c(\{f,[1,0,\lambda]\})$ is $\sP$-hard. Since $\Holant^c(\{f,[1,0,\lambda]\})\redT \Holant^c(\mathcal{F}\cup\{[1,0,\lambda]\})$, the problem $\Holant^c(\mathcal{F}\cup\{[1,0,\lambda]\})$ is also \#P-hard.
\end{proof}

\section{$\ptype$-transformability and Adjacency Condition}
\label{sec:p-transformability}
We start with some simple facts from linear algebra.
Let $M=\begin{bmatrix}
         a_1 & a_2 & \cdots & a_n \\
         b_1 & b_2 & \cdots & b_n
       \end{bmatrix}\ (n\geq 2)$ be a non-negative matrix of rank 2. Then $A=MM^\mathsf{T}=\begin{bmatrix} a & b \\ b & c \end{bmatrix}$ satisfying $a,c>0$. Moreover, by Cauchy-Schwarz inequality, $\det A=ac-b^2>0$.

\begin{lemma}
\label{lemma:unequal-eigenvalues}
    If $a\neq c$ or $b\neq 0$, then $A$ has two distinct positive eigenvalues $\alpha$ and $\beta$.
\end{lemma}
\begin{proof}
    Suppose that $a\neq c$ or $b\neq 0$. The characteristic polynomial of $A$ is $p(t)=t^2-(a+c)t+ac-b^2$. Note that $\Delta=(a+c)^2-4(ac-b^2)=(a-c)^2+4b^2> 0$, so the quadratic equation $p(t)=0$ has two distinct real roots $\alpha$ and $\beta$. Since $\alpha+\beta=a+c>0$ and $\alpha\beta=ac-b^2>0$, both $\alpha$ and $\beta$ are positive.
\end{proof}

The following lemma is a simple case of the Spectral Theorem for real symmetric matrices.
\begin{lemma}
\label{lemma:realsym-diag}
    There is an orthogonal matrix $H$ such that $HAH^\mathsf{T}=\begin{bmatrix}\alpha & 0 \\ 0 & \beta \end{bmatrix}$, where $\alpha$ and $\beta$ are the eigenvalues of $A$.
\end{lemma}

Let $f$ be a non-negative binary function. If $f$ is non-degenerate and affine, then $f=a[1,0,1]$ or $f=a[0,1,0]$ for some $a>0$.

\begin{lemma}
\label{lemma:binary-non-affine}
    Let $f=(a,b,c,d)$ be a non-negative function. Suppose that $f$ is non-degenerate and $f\notin \atype$. Then for any function set $\mathcal{F}$ with $f\in\sig(\mathcal{F})$, $\Holant(\mathcal{F})$ is $\sP$-hard or $\mathcal{F}\subseteq\ttype$ or $\mathcal{F}\subseteq H\ptype$ for some orthogonal matrix $H$.
\end{lemma}
\begin{proof}
    Since $f\in\sig(\mathcal{F})$, the symmetric matrix
    \begin{equation*}
        A=\begin{bmatrix}a & b \\ c & d\end{bmatrix}\begin{bmatrix}a & c \\ b & d\end{bmatrix}=\begin{bmatrix}a^2+b^2 & ac+bd \\ ac+bd & c^2+d^2\end{bmatrix}
    \end{equation*}
    is also realizable. Because $f$ is non-degenerate, $a^2+b^2,c^2+d^2>0$ and $ac+bd\geq 0$. We claim that $ac+bd\neq 0$ or $a^2+b^2\neq c^2+d^2$. Suppose $ac+bd=0$, then $ac=bd=0$ since $f$ is non-negative. So $f=\begin{bmatrix}a & 0 \\ 0 & d\end{bmatrix}$ or $f=\begin{bmatrix}0 & b \\ c & 0\end{bmatrix}$. In both cases, as $f\notin\atype$, $a^2+b^2\neq c^2+d^2$.

    By \autoref{lemma:unequal-eigenvalues} and \autoref{lemma:realsym-diag}, there is some orthogonal matrix $H$ such that $HAH^\mathsf{T}=\begin{bmatrix}\alpha & 0 \\ 0 & \beta \end{bmatrix}$, where $\alpha$ and $\beta$ are the two distinct positive eigenvalues of $A$. Now we perform the transformation $H$ and obtain the following equivalence:
    \begin{equation*}
        \Holant(\{[\alpha,0,\beta]\}\cup H\mathcal{F})\eqredT \Holant(\{A\}\cup\mathcal{F})\eqredT \Holant(\mathcal{F}).
    \end{equation*}
    The latter equivalence follows from the fact $A\in\sig(f)\subseteq\sig(\mathcal{F})$. $\beta/\alpha$ is nonzero and not a root of unity, so if $H\mathcal{F}\not\subseteq \ttype$ and $H\mathcal{F}\not\subseteq \ptype$, the problem is $\sP$-hard by \autoref{thm:Holant-nontrivial-eq}.
\end{proof}

Let $f$ be an irreducible non-negative function of arity $n\geq 2$. For $i\in [n]$, we use $M_i(f)$ to denote the $2\times 2^{n-1}$ matrix whose row is indexed by $x_i\in\{0,1\}$ and columns by $x_1\cdots x_{i-1}x_{i+1}\cdots x_n\in\{0,1\}^{n-1}$, and
\begin{equation*}
    M_i(f)(x_i,x_1\cdots x_{i-1}x_{i+1}\cdots x_n)=f(x_1,...,x_n).
\end{equation*}
We can realize $n$ symmetric matrices with non-negative entries:
\begin{equation*}
    A_i(f)=M_i(f)(M_i(f))^\mathsf{T}=\begin{bmatrix} a_i & b_i \\ b_i & c_i \end{bmatrix},\text{ for }i\in[n].
\end{equation*}
Since $f$ is irreducible, $M_i(f)\ (i\in[n])$ are all of rank 2. Thus for all $i$, $a_i,c_i>0$ and $a_ic_i>b_i^2$.

We say two strings $\mathbf{u},\mathbf{v}\in\{0,1\}^n$ are \emph{adjacent} if $\mathbf{u}\oplus \mathbf{v}$ has Hamming weight 1.

\begin{mydef}[Adjacency Condition]
    A function satisfies the \emph{Adjacency} condition if there are two adjacent elements in its support.
\end{mydef}

\begin{corollary}
\label{coro:adjacency}
    Let $f$ be an irreducible non-negative function of arity $n\geq 2$. If $f$ satisfies the Adjacency condition, then for any function set $\mathcal{F}$ containing $f$, $\Holant(\mathcal{F})$ is $\sP$-hard or $\mathcal{F}\subseteq\ttype$ or $\mathcal{F}\subseteq H\ptype$ for some orthogonal matrix $H$.
\end{corollary}
\begin{proof}
    If $f$ satisfies the Adjacency condition, then there exists some $A_i(f)$ with $b_i\neq 0$. And the conclusion follows from \autoref{lemma:binary-non-affine}.
\end{proof}

\autoref{lemma:binary-non-affine} and \autoref{coro:adjacency} are important for later proofs. It often serves as the first step, to filter some signature sets that are $\ptype$-transformable or lead to \#P-hardness. Those passing this filter must satisfy certain structural properties.

\section{On Special Functions of Arity 4}
\label{sec:special-4}

In this section, we consider some special functions of arity 4.

\begin{lemma}
\label{lemma:special-4}
    Let $f$ be a function of arity $4$, whose signature matrix has the form
    \begin{equation}
    \label{eq:special-4}
        M_f= \begin{bmatrix}
            f_{0000} & f_{0001} & f_{0010} & f_{0011} \\
            f_{0100} & f_{0101} & f_{0110} & f_{0111} \\
            f_{1000} & f_{1001} & f_{1010} & f_{1011} \\
            f_{1100} & f_{1101} & f_{1110} & f_{1111} \\
            \end{bmatrix}
        = \begin{bmatrix}
          1 & 0 & 0 & a \\
          0 & b & c & 0 \\
          0 & c & b & 0 \\
          a & 0 & 0 & 1
        \end{bmatrix}
    \end{equation}
where $a,b,c\geq 0$ and at least two of them are positive. Then $\Holant(f)$ is $\sP$-hard if $f\neq [1,0,1,0,1]$.
\end{lemma}

Before proving the hardness, we introduce the redundant matrices defined by Cai, Guo and Williams \cite{CGW2013}, and a related complexity result. A $4\times 4$ matrix is called \emph{redundant}, if it has identical middle two rows and identical middle two columns. Given a $4\times 4$ redundant matrix $M$, its \emph{compressed matrix}, denoted by $\widetilde{M}$, is the $3\times 3$ matrix $AMB$ where
\begin{equation*}
    A=\begin{bmatrix}
        1 & 0 & 0 & 0 \\
        0 & \frac{1}{2} & \frac{1}{2} & 0 \\
        0 & 0 & 0 & 1
      \end{bmatrix},\
    B=\begin{bmatrix}
        1 & 0 & 0 \\
        0 & 1 & 0 \\
        0 & 1 & 0 \\
        0 & 0 & 1
      \end{bmatrix}.
\end{equation*}

\begin{lemma}[\cite{CGW2013}]
\label{lemma:redundant-matrix}
    Let $f$ be an arity $4$ signature with complex weights. If $M_f$ is redundant and $\widetilde{M_f}$ is nonsingular, then $\Holant(f)$ is $\sP$-hard.
\end{lemma}

Let $f$ be a function whose signature matrix has the form (\ref{eq:special-4}).
For hardness, it is natural to use $f$ to construct a $4\times 4$ redundant signature matrix whose compressed matrix is nonsingular. Consider the gadget in \autoref{fig:tetrahedron}. The following observations can save our labor in calculating the signature $g$ of this gadget:
\begin{enumerate}[label=\textbullet,labelindent=\parindent,leftmargin=*]
  \item $f$ has even support: For all $\mathbf{x}\in\{0,1\}^4$ of odd Hamming weights, $f(\mathbf{x})=0$. Therefore, so does $g$.
  \item For all $\mathbf{x}\in\{0,1\}^4$, $f(\mathbf{x})=f(\overline{\mathbf{x}})$. Thus $g$ also has this property.
  \item The gadget is rotationally symmetric, hence $g_{0101}=g_{0110}$.
\end{enumerate}
\begin{figure}[h]
  \centering
  \includegraphics[scale=0.6]{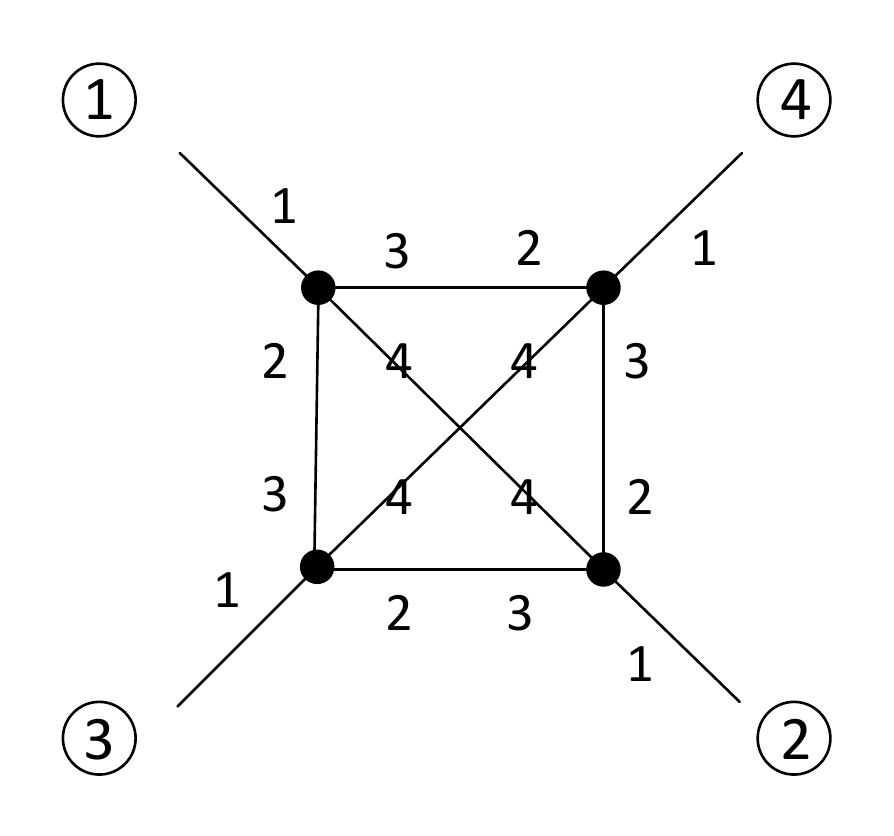}
  \caption{The four black circles are all assigned the function $f$. The inputs of this gadget are numbered 1,3,2,4 counterclockwise.}\label{fig:tetrahedron}
\end{figure}

To summarize, the matrix form of $g$ is
\begin{equation*}
    M_g=\begin{bmatrix}
      x & 0 & 0 & y \\
      0 & z & z & 0 \\
      0 & z & z & 0 \\
      y & 0 & 0 & x
    \end{bmatrix}.
\end{equation*}
Calculation shows that $x=1+4abc+2a^2b^2+c^4$, $y=2c^2+4abc+a^4+b^4$ and $z=2(a^2+b^2)c+2ab(1+c^2)$. Since $a,b,c\geq 0$ and at least two of them are positive, $z\neq 0$. Therefore, $\det \widetilde{M_g}=0$ if and only if $x=y$ if and only if $(1-c^2)^2=(a^2-b^2)^2$.

\begin{proof}[Proof of \autoref{lemma:special-4}]
    Suppose that $f\neq [1,0,1,0,1]$. Then at least one of $a,b,c$ is not $1$. Note that the following two signatures are realizable, by permuting the inputs of $f$:
    \begin{equation*}
        f_b=\begin{bmatrix}
          1 & 0 & 0 & b \\
          0 & a & c & 0 \\
          0 & c & a & 0 \\
          b & 0 & 0 & 1
        \end{bmatrix},
        f_c=\begin{bmatrix}
          1 & 0 & 0 & c \\
          0 & b & a & 0 \\
          0 & a & b & 0 \\
          c & 0 & 0 & 1
        \end{bmatrix}.
    \end{equation*}
    Thus, without loss of generality, we assume that $a\neq 1$. Further, we suppose that $0\leq a<1$ and $0\leq c\leq b$. If it is not the case, we can realize the following signature
    \begin{equation*}
        (M_f)(M_f)^{\mathsf{T}}=(1+a^2)\begin{bmatrix}
                  1 & 0 & 0 & \frac{2a}{1+a^2} \\
                  0 & \frac{b^2+c^2}{1+a^2} & \frac{2bc}{1+a^2} & 0 \\
                  0 & \frac{2bc}{1+a^2} & \frac{b^2+c^2}{1+a^2} & 0 \\
                  \frac{2a}{1+a^2} & 0 & 0 & 1
                \end{bmatrix}.
    \end{equation*}
    Then $0\leq 2a<1+a^2$ since $a\neq 1$. And $0\leq 2bc \leq b^2+c^2$.

    As shown above, we can realize a redundant matrix $M_g$. If $\widetilde{M_g}$ is nonsingular, then $\Holant(g)$ is \#P-hard by \autoref{lemma:redundant-matrix}. Since $\Holant(g)\redT\Holant(f)$, the latter problem is also \#P-hard and we are done. Now suppose that $\det \widetilde{M_g}=0$, which implies that $(1-c^2)^2=(a^2-b^2)^2$. Also, by symmetry, we can use $f_c$ to realize another redundant matrix $M_{h}$ such that $\det\widetilde{M_h}=0$ if and only if $(1-a^2)^2=(b^2-c^2)^2$. Again, we suppose that $\widetilde{M_h}$ is singular otherwise $\Holant(f)$ is \#P-hard. Since $0\leq a<1$ and $0\leq c\leq b$, $(1-a^2)^2=(b^2-c^2)^2$ implies that $1-a^2=b^2-c^2$. Hence we have
    \begin{align*}
      0=(1-c^2)^2-(a^2-b^2)^2 &= (1-c^2-a^2+b^2)(1-c^2+a^2-b^2)  \\
      &=2(1-a^2)(1-c^2+a^2-b^2),
    \end{align*}
    which means that $1-c^2+a^2-b^2=0$. Together with $1-a^2=b^2-c^2$, we obtain $b=1$ and $a=c$. Since at least two of $a,b,c$ are positive, it holds that $0<a<1$. Now it suffices to show that $\Holant(F)$ is \#P-hard, where
    \begin{equation*}
        M_F=\begin{bmatrix}
          1 & 0 & 0 & 1 \\
          0 & a & a & 0 \\
          0 & a & a & 0 \\
          1 & 0 & 0 & 1
        \end{bmatrix}.
    \end{equation*}
    Let $H=\frac{1}{\sqrt{2}}\begin{bsmallmatrix}1 & 1 \\ 1 & -1\end{bsmallmatrix}$, then $\Holant(H^{\otimes 4}F)\redT\Holant(F)$ where
    \begin{equation*}
        M_{H^{\otimes 4}F}=\begin{bmatrix}
                             1+a & 0 & 0 & 1-a \\
                             0 & 0 & 0 & 0 \\
                             0 & 0 & 0 & 0 \\
                             1-a & 0 & 0 & 1+a
                           \end{bmatrix}.
    \end{equation*}
    Since the submatrix $\begin{bsmallmatrix} 1+a & 1-a \\ 1-a & 1+a \end{bsmallmatrix}$ is of full rank, we may apply \autoref{lemma:interpolate-=_4} to interpolate the equality function $=_4$. Thus $\CSP^2(H^{\otimes 4}F)\redT \Holant(F)$. Note that $\CSP([1+a,1-a,1+a])\redT \CSP^2(H^{\otimes 4}F)$, and $\CSP([1+a,1-a,1+a])$ is $\sP$-hard by \autoref{thm:csp-dichotomy}. So $\Holant(F)$ is $\sP$-hard.
\end{proof}

We prove a dichotomy for function sets that contain certain functions of arity $4$.

\begin{lemma}
\label{lemma:4-block-rank-one}
    Let $f$ be a non-negative function of arity $4$. And $\begin{bmatrix}f_{0000} & f_{0011} \\ f_{1100} & f_{1111}\end{bmatrix}=\begin{bmatrix}a & b \\ b & c\end{bmatrix}$ where $b\neq 0$ and $ac>b^2$. Then for any function set $\mathcal{F}$ containing $f$, $\Holant(\mathcal{F})$ is $\sP$-hard or $\mathcal{F}\subseteq \ttype$ or $\mathcal{F}\subseteq H\ptype$ for some orthogonal matrix $H$.
\end{lemma}
\begin{proof}
    Note that $a,b,c>0$.

    First we consider the case that $f$ is reducible. $f$ can not be degenerate, since the matrix $\begin{bmatrix}f_{0000} & f_{0011} \\ f_{1100} & f_{1111}\end{bmatrix}$ is of full rank. So there is a permutation $\pi\in S_4$, such that $f_\pi = g\otimes h$ where $g$ is irreducible and $\arity(g)\geq \arity(h)>0$. Both $g$ and $h$ are non-negative, and neither of them is identically zero. By \autoref{coro:decomposition-non-vanishing}, $\Holant(\mathcal{F}\cup \{g\})\redT \Holant(\mathcal{F})$. There are two cases:
    \begin{enumerate}[label=\textbullet]
      \item $\arity(g)=3$. In this case, $g_{000}\neq 0$ because $f_{0000}=g_{000}h_{0}>0$. And $g_{111}\neq 0$ because $f_{1111}=g_{111}h_{1}>0$. Furthermore, since $f_{0011}\neq 0$, $g\neq [g_{000},0,0,g_{111}]$. So $g$ satisfies the Adjacency condition, and the conclusion follows from \autoref{coro:adjacency}.
      \item $\arity(g)=2$. We may suppose $g=[\lambda,0,\lambda]$ or $g=[0,\lambda,0]$ for some $\lambda>0$, otherwise we are done by \autoref{lemma:binary-non-affine}. Moreover, $g\neq [0,\lambda,0]$ since $f_{0000}\neq 0$. This implies that
          \begin{equation*}
            a=\lambda h_{00},\ b=\lambda h_{00}=\lambda h_{11},\ c=\lambda h_{11},
          \end{equation*}
          which is impossible since $ac>b^2$.
    \end{enumerate}

    Suppose that $f$ is irreducible. And we can suppose that $f$ does not satisfy the Adjacency condition, otherwise we are done. Then the $4\times 4$ signature matrix of $f$ is
    \begin{equation*}
        M_f=\begin{bmatrix}
            a & 0 & 0 & b \\
            0 & x & y & 0 \\
            0 & z & w & 0 \\
            b & 0 & 0 & c
        \end{bmatrix},
    \end{equation*}
    By connecting two inputs of $f$ via an edge, we can realize the following three binary signatures:
    \begin{equation*}
        [a+b,0,b+c],\ [a+x,0,w+c],\ [a+y,0,z+c].
    \end{equation*}
    If $a\neq c$, then the conclusion holds by \autoref{thm:Holant-nontrivial-eq}. Now suppose that $a=c$. For the same reason, we may assume that $x=w$ and $y=z$. Rewrite the matrix of $f$:
    \begin{equation*}
        M_f=\begin{bmatrix}
            a & 0 & 0 & b \\
            0 & x & y & 0 \\
            0 & y & x & 0 \\
            b & 0 & 0 & a
        \end{bmatrix}.
    \end{equation*}
    Let $h$ denote the matrix $\begin{bmatrix}a & b \\ b & a\end{bmatrix}$. Then $h\notin \atype\cup\ptype$ since $a^2>b^2>0$. If at least one of $x,y$ is positive, then $\Holant(f)$ is $\sP$-hard by \autoref{lemma:special-4}. Otherwise we can interpolate $=_4$ using $f$ by \autoref{lemma:interpolate-=_4}. Thus $\CSP(h)\redT \CSP^2(f)\redT \Holant(f)$, which is $\sP$-hard. Since $f\in\mathcal{F}$,  $\Holant(\mathcal{F})$ is also $\sP$-hard.
\end{proof}

\section{The Dichotomy}
\label{sec:dichotomy}

\subsection{The Block-rank-one Condition Captures the Dichotomy}
\label{subsec:block-rank-one}

Given a function $f$ of arity $n$, we use $f^{[t]}$, for each $t\in [n]$, to denote the function
\begin{equation*}
    f^{[t]}(x_1,...,x_t)=\sum_{x_{t+1},...,x_n\in\{0,1\}} f(x_1,...,x_t,x_{t+1},...,x_n).
\end{equation*}

Recall that Holant problems are read-twice $\CSP$s and every \#CSP instance \emph{defines} a function (\autoref{subsec:CSP}). We adopt the notation in \cite{CC2012}, defining the following set of functions for a given $\mathcal{F}$:
\begin{equation*}
    \mathcal{W}_\mathcal{F}=\{F^{[t]}\myvert F \text{ is a function defined by an instance of }\Holant(\mathcal{F})\text{ and }1\leq t \leq \text{arity of }F\}.
\end{equation*}

Note that the functions in $\mathcal{W}_\mathcal{F}$ are not necessarily realizable from $\mathcal{F}$. The following two lemmas show how $\mathcal{W}_\mathcal{F}$ and $\sig(\mathcal{F})$ are related:
\begin{lemma}
\label{lemma:break-edges}
    Let $f\in\mathcal{W}_\mathcal{F}$ be a function of arity $n$. Then there is a function $g\in \sig(\mathcal{F})$ of arity $2n$, such that for all $x_1,x_2,...,x_n\in\{0,1\}$,
    \begin{equation*}
        f(x_1,x_2,...,x_n)=g(x_1,x_1,x_2,x_2,...,x_n,x_n).
    \end{equation*}
\end{lemma}
\begin{proof}
    Suppose that $F$ is an $n$-ary function defined by an instance of $\Holant(\mathcal{F})$. Let $f_1,...,f_k$ denote the constraint functions (not necessarily distinct) that appear in the instance. And let $g$ denote their tensor product $g=f_1\otimes \cdots \otimes f_k$. Then there is a permutation $\pi$ on $[2n]$, such that for all $x_1,x_2,...,x_n\in\{0,1\}$,
    \begin{equation*}
        F(x_1,x_2,...,x_n)=g_\pi(x_1,x_1,x_2,x_2,...,x_n,x_n).
    \end{equation*}
    Moreover, for all $t\in [n]$,
    \begin{equation*}
        F^{[t]}(x_1,...,x_t)=\sum_{x_{t+1},...,x_n\in\{0,1\}} g_\pi(x_1,x_1,...,x_t,x_t,x_{t+1},x_{t+1},...,x_n,x_n).
    \end{equation*}
    Because $f_1,...,f_k\in\mathcal{F}$, both $g$ and $g_\pi$ are realizable from $\mathcal{F}$. And hence, for all $t\in [n]$, the function $\sum_{x_{t+1},...,x_n\in\{0,1\}} g_\pi(x_1,x_2,...,x_{2t},x_{t+1},x_{t+1},...,x_n,x_n)$ is also realizable.
\end{proof}

Intuitively, the function $g$ in \autoref{lemma:break-edges} is obtained by breaking edges of the signature grid that defines $f$. The following lemma goes in the opposite direction - merging dangling edges.

\begin{lemma}
\label{lemma:function-square-definable}
    For $f\in\sig(\mathcal{F})$, $f^2\in\mathcal{W}_\mathcal{F}$.
\end{lemma}
\begin{proof}
    Since $f\in\sig(\mathcal{F})$, there is some $\mathcal{F}$-gate $\Gamma$ that realizes $f$. Suppose that $f$ has arity $n>0$ and $\Gamma$ has $m$ internal edges. Now take two copies of $\Gamma$. For each $k\in [n]$, we connect the $k$-th inputs of the two $\Gamma$'s. This yields an instance of $\Holant(\mathcal{F})$, which defines a function $F$ of arity $n+2m$. And $f^2=F^{[n]}\in\mathcal{W}_\mathcal{F}$.
\end{proof}

Let $M$ be a non-negative matrix. We say $M$ is \emph{block-rank-one} if every two rows of it are linearly dependent or orthogonal. Given a non-negative function $f$ of arity $n$, we say $f$ is \emph{block-rank-one} if either $n=1$ or the matrix $M_{[n-1]}(f)$ is block-rank-one.

Now we impose a condition on $\mathcal{W}_\mathcal{F}$:
\begin{quote}
    \textbf{Block-rank-one}: All functions in $\mathcal{W}_\mathcal{F}$ are block-rank-one.
\end{quote}

We can classify those function sets that do not satisfy this condition:
\begin{lemma}
\label{lemma:dichotomy-block}
    Let $\mathcal{F}$ be a set of non-negative functions. If $\mathcal{F}$ does not satisfy the Block-rank-one condition, then $\Holant(\mathcal{F})$ is $\sP$-hard or $\mathcal{F}\subseteq \ttype$ or $\mathcal{F}\subseteq H\ptype$ for some orthogonal matrix $H$.
\end{lemma}
\begin{proof}
    Let $f\in\mathcal{W}_\mathcal{F}$ be a function of arity $n$. Then by \autoref{lemma:break-edges}, there is a function $g\in \sig(\mathcal{F})$ of arity $2n$, such that for all $x_1,x_2,...,x_n\in\{0,1\}$,
    \begin{equation*}
        f(x_1,x_2,...,x_n)=g(x_1,x_1,x_2,x_2,...,x_n,x_n).
    \end{equation*}

    Now suppose that $f$ is not block-rank-one. By definition, $n\geq 2$ and the two columns of $M_{[n-1]}(f)$ are linearly independent but not orthogonal. Then the first and the last columns of the matrix $M=M_{[2n-2]}(g)$, $g^{x_{2n-1}=x_{2n}=0}$ and $g^{x_{2n-1}=x_{2n}=1}$, are also linearly independent but not orthogonal. Let $h$ denote the $4\times 4$ matrix $M^{\mathsf T}M$. Then $h_{0011}=h_{1100}>0$ and  $h_{0000}h_{1111}>h_{0011}^2$.
    Since $g\in\sig(\mathcal{F})$, $h$ is also realizable. Thus $\Holant(\mathcal{F}\cup\{h\})\redT \Holant(\mathcal{F})$. By \autoref{lemma:4-block-rank-one}, $\Holant(\mathcal{F})$ is \#P-hard or $\mathcal{F}\subseteq \ttype$ or $\mathcal{F}\subseteq H\ptype$ for some orthogonal matrix $H$.
\end{proof}

Surprisingly, the Block-rank-one condition has \emph{captured} the dichotomy. Later we will prove the crucial lemma below:
\begin{lemma}
\label{lemma:block-implies-affine}
    Let $\mathcal{F}$ be a set of non-negative functions. If $\mathcal{F}$ satisfies the Block-rank-one condition, then $\mathcal{F}\subseteq \atype$ or $\mathcal{F}\subseteq \ptype$.
\end{lemma}

Therefore, if $\mathcal{F}$ satisfies the Block-rank-one condition, then $\Holant(\mathcal{F})$ is computable in polynomial time. So our dichotomy is quite simple and it is decidable in polynomial time \cite{CGW2014}:
\begin{thm}
\label{thm:dichotomy}
    Let $\mathcal{F}$ be a set of non-negative functions. The problem $\Holant(\mathcal{F})$ is computable in polynomial time if $\mathcal{F}$ satisfies one of the following three conditions:
    \begin{enumerate}[labelindent=\parindent,leftmargin=*]
      \item $\mathcal{F}\subseteq \ttype$;
      \item $\mathcal{F}\subseteq \atype$;
      \item $\mathcal{F}\subseteq H\ptype$ for some real orthogonal matrix $H$.
    \end{enumerate}
    Otherwise $\Holant(\mathcal{F})$ is $\sP$-hard.
\end{thm}

The remaining is to prove \autoref{lemma:block-implies-affine}. Let $\mathcal{F}$ be a function set that satisfies the Block-rank-one condition. The condition is a little conceptual, since it is imposed on $\mathcal{W}_\mathcal{F}$, in which the functions are not necessarily realizable. To obtain the structure of $\mathcal{F}$, it is more convenient to consider directly the set $\mathcal{F}$ and the functions realizable from it. So in the next subsection, we will introduce a notion equivalent to the Block-rank-one condition. This notion restricts the function set $\sig(\mathcal{F})$.

\subsection{Balance}
\label{subsec:balance}
We define the notion of \emph{balance} for non-negative Holant problems. The notion was introduced for non-negative $\CSP$ by Cai, Chen and Lu \cite{CCL2011}.
\begin{mydef}[Balance]
    Let $\mathcal{F}$ be a set of non-negative functions. $\mathcal{F}$ is called \emph{balanced} if for any function  $f\in \sig(\mathcal{F})$, every signature matrix in $\{M_{[r]}(f)\myvert 1\leq r \leq \arity(f)\}$ is block-rank-one. A non-negative function $f$ is \emph{balanced} if the set $\{f\}$ is balanced.
\end{mydef}
Note that in the definition above, when $r=\arity(f)$, the matrix $M_{[r]}(f)$ is a column vector and hence trivially block-rank-one.

Balanced sets satisfy the Block-rank-one condition. Generally, we have the following lemma.

\begin{lemma}
\label{lemma:definable-balanced}
    Let $\mathcal{F}$ be a set of non-negative functions. Suppose that $\mathcal{F}$ is balanced. Then for any $f\in\mathcal{W}_\mathcal{F}$, every matrix in $\{M_{[r]}(f)\myvert 1\leq r \leq \arity(f)\}$ is block-rank-one.
\end{lemma}
\begin{proof}
    Let $f\in\mathcal{W}_\mathcal{F}$ be a function of arity $n$. Then by \autoref{lemma:break-edges}, there exists a function $g\in\sig(\mathcal{F})$ of arity $2n$, such that for all $x_1,x_2,...,x_n\in\{0,1\}$,
    \begin{equation*}
        f(x_1,x_2,...,x_n)=g(x_1,x_1,x_2,x_2,...,x_n,x_n).
    \end{equation*}
    Therefore, for any $r\in [n]$, $M_{[r]}(f)$ is a submatrix of $M_{[2r]}(g)$. Because $\mathcal{F}$ is balanced, $M_{[2r]}(g)$ is block-rank-one. Hence so is $M_{[r]}(f)$.
\end{proof}

Let $f$ be a non-negative function of arity $n$. And let $s_1,...,s_n$ be $n$ non-negative unary functions. We call $(s_1,...,s_n)$ a \emph{vector representation} of $f$ if for all $\mathbf{x}\in\{0,1\}^n$, either $f(\mathbf{x})=0$ or $f(\mathbf{x})=s_1(x_1)\cdots s_n(x_n)$.
\begin{lemma}[\cite{CCL2011}]
\label{lemma:vector-representation}
    Let $f$ be a non-negative function of arity $n$. If $f^{[t]}$ is block-rank-one for all $t\in [n]$, then $f$ has a vector representation.
\end{lemma}

\begin{lemma}
\label{lemma:realizable-representation}
    Let $\mathcal{F}$ be a set of non-negative functions that satisfies the Block-rank-one condition. Then every function in $\sig(\mathcal{F})$ has a vector representation.
\end{lemma}
\begin{proof}
    Let $f$ be a function in $\sig(\mathcal{F})$ of arity $n$. By \autoref{lemma:function-square-definable}, $f^2\in \mathcal{W}_\mathcal{F}$. Then $f^2$ has a vector representation $(s_1,...,s_n)$ by \autoref{lemma:vector-representation}. Let $(s_1',...,s_n')$ be $n$ non-negative unary functions such that for all $i\in [n]$, $s_i'(a)=\sqrt{s_i(a)}$ for $a\in\{0,1\}$. Then $(s_1',...,s_n')$ is a vector representation of the function $f$.
\end{proof}

Now we are able to prove the equivalence between the notion of balance and the Block-rank-one condition.

\begin{lemma}
\label{lemma:balance-eq-block}
    Let $\mathcal{F}$ be a set of non-negative functions. $\mathcal{F}$ is balanced if and only if $\mathcal{F}$ satisfies the Block-rank-one condition.
\end{lemma}
\begin{proof}
    The necessity follows directly from \autoref{lemma:definable-balanced}. We only need to show the sufficiency.

    Let $f$ be an $n$-ary function in $\sig(\mathcal{F})$, with $n\geq 2$. And suppose that $M=M_{[r]}(f)$ is not block-rank-one for some $r\in [n]$. Then there exist two rows of $M$, indexed by some $\mathbf{x},\mathbf{y}\in\{0,1\}^r$, which are linearly independent but not orthogonal. So we can realize a signature $g=MM^{\mathsf T}$. Its submatrix
    \begin{equation*}
        h=\begin{bmatrix}
            g(\mathbf{x},\mathbf{x}) & g(\mathbf{x},\mathbf{y}) \\
            g(\mathbf{y},\mathbf{x}) & g(\mathbf{y},\mathbf{y})
        \end{bmatrix} =
        \begin{bmatrix}
          a & b \\
          b & c
        \end{bmatrix}
    \end{equation*}
    is of full rank and $a,b,c>0$. But by \autoref{lemma:realizable-representation}, $g$ has a vector representation $(s_1,...,s_{2r})$, such that for all $\mathbf{u}\in\supp(g)$, $g(\mathbf{u})=s_1(u_1)\cdots s_{2r}(u_{2r})$. Let $s=s_1\otimes \cdots \otimes s_r$ and $t=s_{r+1}\otimes \cdots \otimes s_{2r}$. Then
    \begin{equation*}
        h=\begin{bmatrix}
            s(\mathbf{x})t(\mathbf{x}) & s(\mathbf{x})t(\mathbf{y}) \\
            s(\mathbf{y})t(\mathbf{x}) & s(\mathbf{y})t(\mathbf{y})
        \end{bmatrix},
    \end{equation*}
    which is singular. A contradiction.
\end{proof}

Having shown the equivalence, we turn to consider some properties of balanced sets. There are two basic facts about balance. Later we will often use them but without explicit reference.
\begin{lemma}
    If $\mathcal{F}\subseteq \mathcal{G}$ and $\mathcal{G}$ is balanced, then $\mathcal{F}$ is also balanced.
\end{lemma}
\begin{lemma}
    If $f\in\sig(\mathcal{F})$ and $\mathcal{F}$ is balanced, then $\mathcal{F}\cup\{f\}$ is also balanced.
\end{lemma}

In Boolean $\CSP$, the two unary functions $[1,0]$ and $[0,1]$ can be simulated \cite{DGJ2009}. And the function $[1,1]$ is the unary equality function, which is freely available. These unary functions make it more convenient to construct certain functions. But in Holant problems, generally we do not know how to realize or simulate them. Fortunately, we can circumvent this difficulty by the lemma below. It follows from \autoref{lemma:balance-pinning} and \autoref{lemma:balance-[1,1]}.
\begin{lemma}
\label{lemma:balance-unary}
    If $\mathcal{F}$ is balanced, then the set $\mathcal{F}\cup\{[1,0],[0,1],[1,1]\}$ is balanced.
\end{lemma}

\begin{lemma}
\label{lemma:balance-pinning}
    If $\mathcal{F}$ is balanced, then $\mathcal{F}\cup\{[1,0],[0,1]\}$ is balanced.
\end{lemma}
\begin{proof}
    Let $f$ be a function in $\sig(\mathcal{F}\cup\{[1,0],[0,1]\})$ and let $M_{[r]}(f)$ be any signature matrix. For each row of $M_{[r]}(f)$, if it is a zero vector, we delete it. And then we delete all the columns of zeros. Let $M$ denote the remaining submatrix of $M_{[r]}(f)$, then $M$ is block-rank-one if and only if $M_{[r]}(f)$ is block-rank-one. Moreover, there exits a function $F\in\sig(\mathcal{F})$ of arity $m$ such that $M$ is a submatrix of $M_{[s]}(F)$ for some $s\in [m]$. Since $\mathcal{F}$ is balanced, $M_{[s]}(F)$ is block-rank-one. Thus so is $M$.
\end{proof}

\begin{lemma}
\label{lemma:realizable-summing}
    Suppose that $\mathcal{F}$ is a balanced set of non-negative functions. Let $f$ be an $n$-ary function in $\sig(\mathcal{F})$ and let $F$ denote the function $f^2$. Then for each $t\in [n]$, there exists a constant $\lambda_t>0$ such that $F^{[t]}=\lambda_t (f^{[t]})^2$.
\end{lemma}
\begin{proof}
    Impose induction on $t$. The base case $t=n$ is trivial where $\lambda_n=1$.

    Suppose that $F^{[t]}=\lambda_{t} (f^{[t]})^2$ for $t=k+1\leq n$. Consider the case $t=k$. For all $\mathbf{x}\in\{0,1\}^k$,
    \begin{align*}
        F^{[k]}(\mathbf{x}) & = F^{[k+1]}(\mathbf{x},0) + F^{[k+1]}(\mathbf{x},1) \\
            & = \lambda_{k+1} \left[\left(f^{[k+1]}(\mathbf{x},0)\right)^2+\left(f^{[k+1]}(\mathbf{x},1)\right)^2\right].
    \end{align*}
    Note that the function $F^{[k+1]}\in \mathcal{W}_{\mathcal{F}}$ since $F=f^2\in \mathcal{W}_\mathcal{F}$. Because $\mathcal{F}$ is balanced, $F^{[k+1]}$ is block-rank-one by \autoref{lemma:balance-eq-block}. Thus the function $f^{[k+1]}=\sqrt{F^{[k+1]}/\lambda_{k+1}}$ is also block-rank-one, which implies that the two column vectors of the matrix $M_{[k]}(f^{[k+1]})$, denoted by $\mathbf{v}_0$ and $\mathbf{v}_1$, are orthogonal or linearly dependent:
    \begin{enumerate}[label=(\arabic*),labelindent=\parindent,leftmargin=*]
      \item $\mathbf{v}_0$ and $\mathbf{v}_1$ are orthogonal. Then for all $\mathbf{x}\in\{0,1\}^k$,
        \begin{align*}
          F^{[k]}(\mathbf{x}) &= \lambda_{k+1} \left[\left(f^{[k+1]}(\mathbf{x},0)\right)^2+\left(f^{[k+1]}(\mathbf{x},1)\right)^2\right]\\
            &= \lambda_{k+1} \left(f^{[k+1]}(\mathbf{x},0)+f^{[k+1]}(\mathbf{x},1)\right)^2 \\
            &= \lambda_{k+1} \left(f^{[k]}(\mathbf{x})\right)^2.
        \end{align*}
      \item $\mathbf{v}_0$ and $\mathbf{v}_1$ are linearly dependent. Without loss of generality, we assume that $\mathbf{v}_1=\lambda \mathbf{v}_0$ for some $\lambda\geq 0$. Then for all $\mathbf{x}\in\{0,1\}^k$,
          \begin{align*}
            F^{[k]}(\mathbf{x}) = \lambda_{k+1}(1+\lambda^2) \left(f^{[k+1]}(\mathbf{x},0)\right)^2
                = \lambda_{k+1}\frac{1+\lambda^2}{(1+\lambda)^2} \left(f^{[k]}(\mathbf{x})\right)^2.
          \end{align*}
    \end{enumerate}
    In either case, the conclusion holds. This completes the induction.
\end{proof}

\begin{lemma}
\label{lemma:balance-[1,1]}
    If $\mathcal{F}$ is balanced, then $\mathcal{F}\cup\{[1,1]\}$ is balanced.
\end{lemma}
\begin{proof}

    Suppose that $[1,1]\notin\sig(\mathcal{F})$, otherwise we are done. Let $g$ be an $n$-ary function in $\sig(\mathcal{F}\cup\{[1,1]\})$. We need to show that all the matrices in $\{M_{[r]}(g)\myvert 1\leq r\leq \arity(g)\}$ are block-rank-one.

    Let $\Gamma$ denote the gate that realizes $g$. If there is an isolated vertex with a dangling edge in $\Gamma$, assigned the function $[1,1]$, then we remove this vertex; If there are two adjacent vertices, both assigned the function $[1,1]$, then we delete the pair. Repeat removing until no such vertices. Finally we obtain a new gate $\Gamma'$. If $\Gamma'$ has no dangling edges, then we are done. Suppose not. Let $h$ denote the function that $\Gamma'$ realizes. And for all $x_1,...,x_n\in\{0,1\}$, $g(x_1,...,x_n)=2^sh(x_{i_1},...,x_{i_t})$ where $1\leq i_1<\cdots <i_t\leq n$ and $s$ denotes the number of pairs we delete. It suffices to prove that the signature matrices of $h$ are all block-rank-one.

    Note that $h=f^{[t]}$ for some $f\in\sig(\mathcal{F})$ and $1\leq t\leq \arity(f)$. Let $F$ denote the function $f^2$. Then by \autoref{lemma:realizable-summing}, there is a constant $\lambda_t>0$ such that $F^{[t]}=\lambda_t (f^{[t]})^2$. Therefore, for any $r\in [t]$, the two matrices $M_{[r]}(f^{[t]})$ and $M_{[r]}(F^{[t]})$ are both block-rank-one or neither. Since $F^{[t]}\in\mathcal{W}_\mathcal{F}$, all of its signature matrices are block-rank-one by \autoref{lemma:definable-balanced}. Thus every matrix in $\{M_{[r]}(f^{[t]})\myvert 1\leq r \leq t\}$ is block-rank-one.
\end{proof}

Now we make some preparations for the proof of \autoref{lemma:block-implies-affine}. Let $g$ be a function whose signature has the form $[a,0,a,0]$ or $[a,0,...,0,b]\ (a,b>0\text{ and }\arity(g)\geq 3)$. We will show that, if $\mathcal{F}\cup\{g\}$ is balanced, then $\mathcal{F}\subseteq\atype$ or $\mathcal{F}\subseteq\ptype$. In the next subsection, we will see that, with a trivial exception, such a function $g$ is realizable from $\mathcal{F}\cup\{[1,0],[0,1],[1,1]\}$ if $\mathcal{F}$ is balanced.

We use $\mathbf{0}$ to denote a string of $0$'s. Its length will be clear from the context.

\begin{lemma}
\label{lemma:balance-ternary-xor}
    Let $\mathcal{F}$ be a set of non-negative functions and let $g=[1,0,1,0]$. If $\mathcal{F}\cup\{g\}$ is balanced, then $\mathcal{F}\subseteq \atype$.
\end{lemma}
\begin{proof}
    Suppose that $\mathcal{F}\cup\{g\}$ is balanced. Then $\mathcal{G}=\mathcal{F}\cup\{g,[1,0],[0,1],[1,1]\}$ is balanced by \autoref{lemma:balance-unary}. Indeed, we will prove that $\sig(\mathcal{G})\subseteq \atype$. Then the conclusion follows directly since $\mathcal{F}\subseteq\sig(\mathcal{G})$.

    We say a function is \emph{pure}, if it has range $\{0,\lambda\}$ for some $\lambda>0$. First we show that all the functions in $\sig(\mathcal{G})$ are pure. Let $f\in \sig(\mathcal{G})$ be a function of arity $n$ that is not identically zero. Note that $g^{x_1=1}$ is the disequality function $[0,1,0]$, so if $f(\mathbf{0})=0$ then we can flip some inputs of $f$ to obtain a function $f'$ such that $f'(\mathbf{0})\neq 0$. $f'$ is pure if and only if $f$ is pure. Therefore, we can assume $f(\mathbf{0})=1$ and then show that for all $\mathbf{x}\in\supp(f)$, $f(\mathbf{x})=1$.

    For contradiction, suppose that the set $S=\{\mathbf{x}\in\supp(f)\myvert f(\mathbf{x})\neq 1\}$ is nonempty. Let $\mathbf{u}$ be an element of $S$ that has minimum Hamming weight. We define $I=\{k\in [n]\myvert \text{the $k$th bit of }\mathbf{u}\text{ is }0\}$. Then we can obtain the signature $h=[1,0,...,0,\lambda]\ (\lambda=f(\mathbf{u}))$ by pinning the inputs of $f$ indexed by $I$ to $0$. Further, by connecting $\arity(h)-1$ copies of $[1,1]$ to $h$, we get a function $h'=[1,\lambda]$. And connecting $h'$ with an input of $g$ yields a signature matrix $\begin{bmatrix}1 & \lambda \\ \lambda & 1\end{bmatrix}$, which is not block-rank-one.

    Therefore, all the functions in $\sig(\mathcal{F})$ must be pure. It follows that every unary function in $\sig(\mathcal{G})$ is affine. Based on this, we show by induction on function arity $n\ (n\geq 2)$ that $\sig(\mathcal{G})\subseteq \atype$.

    Suppose that all the $(n-1)$-ary functions in $\sig(\mathcal{G})$ are affine. And let $f\in \sig(\mathcal{G})$ be a function of arity $n$. By the induction hypothesis, for all $i\in [n]$, $f^{x_i=0}$ and $f^{x_i=1}$ are both affine. Moreover, the following realizable function of arity $n-1$ is also affine:
    \begin{equation*}
        h(x_1,...,x_{n-1})=\sum_{y,z\in\{0,1\}} g(x_1,y,z)f(y,z,x_2,...,x_{n-1}).
    \end{equation*}
    Suppose that $f$ does not have affine support. Again, we may assume that $f(\mathbf{0})>0$. Then there exist two elements $\mathbf{a}=a_1\cdots a_n,\mathbf{b}=b_1\cdots b_n \in \supp(f)$ such that $\mathbf{a}\oplus\mathbf{b}=\mathbf{c}\notin \supp(f)$. There is some $i\in [n]$ such that $a_i\neq b_i$. Without loss of generality, we assume that $a_1=0$ and $b_1=1$. For convenience, let $\mathbf{a}=a_1a_2a'$, $\mathbf{b}=b_1b_2b'$ and $\mathbf{c}=c_1c_2c'$. Note that
    \begin{align*}
        h(a_1\oplus a_2,a') &= f(a_1,a_2,a') + f(\overline{a_1},\overline{a_2},a')\neq 0, \\
        h(b_1\oplus b_2,b') &= f(b_1,b_2,b') + f(\overline{b_1},\overline{b_2},b')\neq 0.
    \end{align*}
    $h$ is affine and $h(\mathbf{0})\geq f(\mathbf{0})>0$, so $h(c_1\oplus c_2,c')= f(c_1,c_2,c') + f(\overline{c_1},\overline{c_2},c')=f(\overline{c_1},\overline{c_2},c')\neq 0$. Since $a_1=0$ and $b_1=1$, $\overline{c_1}=0$. Thus $a_2a',\overline{c_2}c'\in \supp(f^{x_1=0})$, which is affine. Note that $f(\mathbf{0})>0$, the support of $f^{x_1=0}$ is indeed a linear space. Therefore,  $(a_2a')\oplus(\overline{c_2}c')=\overline{b_2}b'\in\supp(f^{x_1=0})$. This implies that $\overline{b_1b_2}b'\in \supp(f)$. Because $f$ is pure, we have
     \begin{align*}
        h(b_1\oplus b_2,b') &= f(b_1,b_2,b') + f(\overline{b_1},\overline{b_2},b') = 2f(\mathbf{0}), \\
        h(c_1\oplus c_2,c') &= f(\overline{c_1},\overline{c_2},c') = f(\mathbf{0}).
     \end{align*}
    This contradicts the fact that $h$ is affine. Therefore, $f$ has affine support. And we complete the induction.
\end{proof}

Dyer, Goldberg and Jerrum proved a lemma on non-negative functions:

\begin{lemma}[\cite{DGJ2009}]
    Suppose that $f$ does not have affine support. Then $\CSP(\{f\})$ is $\sP$-hard.
\end{lemma}
Indeed, the proof of this lemma shows inductively the following lemma.
\begin{lemma}
\label{lemma:nonaffine-supp}
    Let $f$ be a non-negative function that does not have affine support. Then there is a binary function $h\in\sig(\{f,=_3,[1,0],[0,1],[1,1]\})$ whose support is not affine.
\end{lemma}

\begin{lemma}
\label{lemma:balance-eq}
    Let $\mathcal{F}$ be a set of non-negative functions and let $g=[a,0,...,0,b]$ be a general equality function where $\arity(g)\geq 3$ and $a,b>0$. If $\mathcal{F}\cup\{g\}$ is balanced, then $\mathcal{F}\subseteq\atype$ or $\mathcal{F}\subseteq\ptype$.
\end{lemma}
\begin{proof}
    Since $\mathcal{F}\cup\{g\}$ is balanced, so is $\mathcal{F}\cup\{g,[1,0],[0,1],[1,1]\}$ by \autoref{lemma:balance-unary}.

    First we show that every function in $\mathcal{F}$ has affine support. For contradiction, suppose that a function $f\in \mathcal{F}$ does not have affine support.  By \autoref{lemma:nonaffine-supp}, there is a binary function $h\in\sig(\{f,=_3,[1,0],[0,1],[1,1]\})$ such that the support of $h$ is not affine. Now using $g$ and $[1,1]$, we can realize a ternary function $g'=[a,0,0,b]$. Let $\Gamma$ denote the gate that realizes $h$. We replace each equality function $=_3$ in $\Gamma$ by $g'$. This substitution produces a binary function $h'\in\sig(\{f,g,[1,0],[0,1],[1,1]\})$ whose support is the same as that of $h$ and hence is not affine. Written as a $2\times 2$ matrix, $h'$ has exactly one zero entry, not block-rank-one. This contradicts the fact that the set $\mathcal{F}\cup\{g,[1,0],[0,1],[1,1]\}$ is balanced.

    Now suppose that $\mathcal{F}\not\subseteq\atype$. Then $\mathcal{F}$ can not be pure and hence, as in the proof of \autoref{lemma:balance-ternary-xor}, we can realize an unary function $[x,y]$ with $x,y>0$ and $x\neq y$. We show that $\mathcal{F}\subseteq\ptype$. For contradiction, suppose that there exists a function $f\in\mathcal{F}\backslash\ptype$. Then by \autoref{lemma:non-ptype}, there are two cases:
    \begin{enumerate}[label=\textbullet]
      \item There is a ternary function $h(x_1,x_2,x_3)\in\sig(\{f,[1,0],[0,1],[1,1]\})$ whose support is determined by a linear equation $x_1\oplus x_2\oplus x_3=c$ for some $c\in\{0,1\}$. We only consider the case $c=0$, the other case is similar. With $[x,y]$ and $[1,1]$ at hand, we can realize two binary functions:
          \begin{align*}
            &h_1 =[x,y]M_{[1]}(h)=(xh_{000},yh_{101},yh_{110},xh_{011}), \\
            &h_2 =[1,1]M_{[1]}(h)=(h_{000},h_{101},h_{110},h_{011}).
          \end{align*}
          At least one of them is not block-rank-one.
      \item There is a binary function $h=(a,b,c,d)\in \sig(\{f,[1,0],[0,1],[1,1]\})$ where $abcd\neq 0$ and $h\notin\ptype$. $h$ is not block-rank-one.
    \end{enumerate}
    In either case, we can realize a binary function that is not block-rank-one. This is impossible because $\mathcal{F}\cup\{g,[1,0],[0,1],[1,1]\}$ is balanced.
\end{proof}

\subsection{Proof of \autoref{lemma:block-implies-affine}}
\label{subsec:block-implies-affine}
Now we are ready to prove \autoref{lemma:block-implies-affine}.

Let $\mathcal{F}$ be a set of non-negative functions that satisfies the Block-rank-one condition. \autoref{lemma:balance-eq-block} shows that $\mathcal{F}$ is balanced. And by \autoref{lemma:balance-unary}, the set
\begin{equation*}
    \mathcal{G}=\mathcal{F}\cup\{[1,0],[0,1],[1,1]\}
\end{equation*}
is also balanced. So it suffices to prove that $\mathcal{G}\subseteq\atype$ or $\mathcal{G}\subseteq \ptype$.

First we consider the case $\mathcal{G}\subseteq \ttype$. In this case, every nondegenerate binary function in $\sig(\mathcal{G})$ has the form $[a,0,b]$ or $(0,a,b,0)$. Thus all of them are of product type. Since the set $\ptype$ is closed under tensor product, $\mathcal{G}\subseteq \ptype$.

Now suppose that $\mathcal{G}\not\subseteq \ttype$. Then there are a function $F\in\mathcal{G}$ and a permutation $\pi$ such that $F_\pi=F_1\otimes F_2$ where $F_1,F_2$ are both non-negative functions ($F_2$ is absent if $F$ is irreducible) and $F_1$ is irreducible with arity $n\geq 3$. Since $\{[1,0],[0,1]\}\subset \mathcal{G}$ and $F_2$ is not identically zero, by pinning we can realize an irreducible function $f=cF_1$ for some $c>0$.

For $1\leq i<j\leq n$ and $a,b\in\{0,1\}$, we use $f_{ij}^{ab}$ denote the column vector $M_{[n-2]}(f^{x_i=a,x_j=b})$. And we define the $2^{n-2}\times 2^2$ matrices $M_{ij}=(f_{ij}^{00},f_{ij}^{01},f_{ij}^{10},f_{ij}^{11})$. Note that $f$ can not satisfy the Adjacency condition, otherwise some $M_i(f)$ (see \autoref{sec:p-transformability}) is not block-rank-one. So we have
\begin{align*}
    \innerprod{f_{ij}^{00},f_{ij}^{01}}=0,\ \innerprod{f_{ij}^{00},f_{ij}^{10}}=0, \\
    \innerprod{f_{ij}^{11},f_{ij}^{01}}=0,\ \innerprod{f_{ij}^{11},f_{ij}^{10}}=0,
\end{align*}
where $\innerprod{\cdot,\cdot}$ denotes the inner product. Therefore, for every pair $(i,j)$, the $4\times 4$ matrix $B_{ij}=(M_{ij})^{\mathsf{T}}M_{ij}$ has the form
\begin{equation*}
    \begin{bmatrix}
      a & 0 & 0 & b \\
      0 & x & y & 0 \\
      0 & y & z & 0 \\
      b & 0 & 0 & c
    \end{bmatrix}
\end{equation*}
where $B_{ij}(\mathbf{u},\mathbf{v})=\innerprod{f_{ij}^{\mathbf{u}},f_{ij}^{\mathbf{v}}}$ for all $\mathbf{u},\mathbf{v}\in\{0,1\}^2$. By Cauchy-Schwarz inequality, $ac\geq b^2$ and $xz\geq y^2$.

\begin{lemma}
    If $B_{ij}$ is diagonal for all $1\leq i<j\leq n$, then $\mathcal{G}\subseteq\atype$ or $\mathcal{G}\subseteq\ptype$.
\end{lemma}
\begin{proof}
    Suppose that $B_{ij}$ is diagonal for all $i<j$. Then for any two different elements $\mathbf{x},\mathbf{y}\in\supp(f)$, their bitwise XOR $\mathbf{x}\oplus\mathbf{y}$ is of Hamming weight $\geq 3$. Let $\mathbf{u}=u_1\cdots u_n$ and $\mathbf{v}=v_1\cdots v_n$ be two different elements such that $\mathbf{u}\oplus\mathbf{v}$ has minimum Hamming weight $m$:
    \begin{equation*}
        m=\min
        \{w(\mathbf{a}\oplus \mathbf{b})\myvert \mathbf{a},\mathbf{b}\in \supp(f)\text{ and } \mathbf{a}\neq\mathbf{b}\}.
    \end{equation*}
    We define two sets: $S_c=\{k\in [n]\myvert u_k=v_k=c\}$ for $c\in\{0,1\}$. Then the function
    \begin{equation*}
        g=\partial_{[1,0]}^{S_0}(\partial_{[0,1]}^{S_1}(f))
    \end{equation*}
    has arity $m\geq 3$ and its support is $\{\mathbf{w},\overline{\mathbf{w}}\}$ for some $\mathbf{w}\in\{0,1\}^m$. Since we can permute the inputs of $g$, it is reasonable to assume that $\mathbf{w}=0^s1^t\ (s+t=m\geq 3)$. Then we can realize two functions $M_{[s]}(g)(M_{[s]}(g))^{\mathsf T}$ and $(M_{[s]}(g))^{\mathsf T}M_{[s]}(g)$, of arity $2s$ and $2t$, respectively. The signatures of the two functions are both of the form $[a,0,...,0,b]$ with $a,b>0$ and at least one of them, say $h$, has arity $\geq 3$. Since $h\in\sig(\mathcal{G})$, $\mathcal{G}\cup\{h\}$ is balanced. So by \autoref{lemma:balance-eq}, $\mathcal{G}\subseteq\atype$ or $\mathcal{G}\subseteq \ptype$.
\end{proof}

Now suppose that some $B_{ij}=\begin{bsmallmatrix}a & 0 & 0 & b \\ 0 & x & y & 0 \\ 0 & y & z & 0 \\ b & 0 & 0 & c\end{bsmallmatrix}$ is not diagonal. Then $a,b,c>0$ or $x,y,z>0$. Let $C_1=\begin{bmatrix}a & b \\ b & c\end{bmatrix}$ and $C_2=\begin{bmatrix}x & y \\ y & z\end{bmatrix}$ be the two submatrices. By the definition of $B_{ij}$, any one of $\{C_1,C_2\}$ can not be the zero matrix, otherwise the other is not block-rank-one since $f$ is irreducible. Moreover, with the unary function $[1,1]$ at hand, we can realize two matrices:
\begin{align*}
  D_1=\partial_{[1,1]}^{\{1\}}(B_{ij})=\begin{bmatrix}
                                         a & y & z & b \\
                                         b & x & y & c
                                       \end{bmatrix},\\
  D_2=\partial_{[1,1]}^{\{2\}}(B_{ij})=\begin{bmatrix}
                                         a & x & y & b \\
                                         b & y & z & c
                                       \end{bmatrix}.
\end{align*}
Again, both $D_1$ and $D_2$ must be block-rank-one. This implies that $a=b=c=x=y=z>0$. By pinning an input of $B_{ij}$ to $0$, we get the function $g=a[1,0,1,0]$. Thus $\mathcal{G}\cup\{g\}$ is balanced, which implies that $\mathcal{G}\subseteq\atype$ by \autoref{lemma:balance-ternary-xor}. This completes the proof.

\section{Back to \#CSP with Non-negative Weights}
\label{sec:back-to-csp}

In \autoref{subsec:balance}, we show the equivalence between the Block-rank-one condition and balance. In fact, using the same method, we can prove that the notions of weak balance and balance in \cite{CCL2011} are equivalent, without assuming $\mathrm{FP}\neq \sP$.

Let $D=\{1,2,...,d\}\ (d>1)$ be a finite domain. For completeness, here we give the definitions of the two notions.

\begin{mydef}[Weak Balance, \cite{CCL2011}]
    We say $\mathcal{F}$ is \emph{weakly balanced} if for any input instance $I$ of $\CSP(\mathcal{F})$ (which defines a non-negative function $F(x_1,...,x_n)$ over $D$) and for any integer $a:1\leq a<n$, the following $d^a\times d$ matrix $M$ is \emph{block-rank-one}: the rows of $M$ are indexed by $\mathbf{u}\in D^a$ and the columns are indexed by $v\in D$, and
    \begin{equation*}
        M(\mathbf{u},v)=\sum_{\mathbf{w}\in D^{n-a-1}} F(\mathbf{u},v,\mathbf{w}),\ \text{for all }\mathbf{u}\in D^a\text{ and }v\in D.
    \end{equation*}
    For the special case when $a+1=n$, we have $M(\mathbf{u},v)=F(\mathbf{u},v)$ is block-rank-one.
\end{mydef}

\begin{mydef}[Balance, \cite{CCL2011}]
    We say $\mathcal{F}$ is \emph{balanced} if for any input instance $I$ of $\CSP(\mathcal{F})$ (which defines a non-negative function $F(x_1,...,x_n)$ over $D$) and for any integers $a,b:1\leq a<b\leq n$, the following $d^a\times d^{b-a}$ matrix $M$ is \emph{block-rank-one}: the rows of $M$ are indexed by $\mathbf{u}\in D^a$ and the columns are indexed by $\mathbf{v}\in D^{b-a}$, and
    \begin{equation*}
        M(\mathbf{u},\mathbf{v})=\sum_{\mathbf{w}\in D^{n-b}} F(\mathbf{u},\mathbf{v},\mathbf{w}),\ \text{for all }\mathbf{u}\in D^a\text{ and }\mathbf{v}\in D^{b-a}.
    \end{equation*}
    For the special case when $b=n$, we have $M(\mathbf{u},\mathbf{v})=F(\mathbf{u},\mathbf{v})$ is block-rank-one.
\end{mydef}

According to the definition of weak balance, \autoref{lemma:vector-representation} has a direct corollary.
\begin{corollary}
\label{coro:vector}
    Let $\mathcal{F}$ be a function set that is weakly balanced. Then for any function $F(x_1,...,x_n)$ defined by an instance of $\CSP(\mathcal{F})$ and any integer $t\in [n]$, $F^{[t]}$ has a vector representation.
\end{corollary}

By definition, balance implies weak balance. Now we show the other direction.
\begin{lemma}
    Suppose that a function set $\mathcal{F}$ is weakly balanced. Then $\mathcal{F}$ is balanced.
\end{lemma}
\begin{proof}
    For contradiction, we assume that $\mathcal{F}$ is not balanced. Then by the definition of balance, there exists an $n$-ary function $F$ defined by some instance $I$ of $\CSP(\mathcal{F})$ and two integers $a,b:1\leq a<b\leq n$, such that the $d^{a}\times d^{b-a}$ matrix $M$ of $F^{[b]}$ is not block-rank-one. That is, there exist two rows of $M$, indexed by $\mathbf{u}_1,\mathbf{u}_2\in D^{a}$, which are linearly independent but not orthogonal.

    Now consider the matrix $G=MM^{\mathsf T}$. For all $\mathbf{x},\mathbf{y}\in D^a$,
    \begin{equation*}
        G(\mathbf{x},\mathbf{y})=\sum_{\mathbf{z}\in D^{b-a}} F^{[b]}(\mathbf{x},\mathbf{z}) F^{[b]}(\mathbf{y},\mathbf{z}).
    \end{equation*}
    $G$ has a $2\times 2$ submatrix
    \begin{equation*}
        g=\begin{bmatrix}
            G(\mathbf{u}_1,\mathbf{u}_1) & G(\mathbf{u}_1,\mathbf{u}_2) \\
            G(\mathbf{u}_2,\mathbf{u}_1) & G(\mathbf{u}_2,\mathbf{u}_2)
        \end{bmatrix}=
        \begin{bmatrix}
          p & q \\
          q & r
        \end{bmatrix}
    \end{equation*}
    where $p,q,r>0$ and $pr>q^2$.

    Suppose that $I$ has variables: $\mathbf{x}=(x_1,...,x_a),\mathbf{z}=(z_{a+1},...,z_{b}),\mathbf{w}=(w_{b+1},...,w_n)$. We add a copy of $I$ on variables: $\mathbf{y}=(y_1,...,y_a),\mathbf{z}=(z_{a+1},...,z_{b}),\mathbf{w}'=(w_{b+1}',...,w_n')$. Then the two instances constitute a new instance, which defines the function
    \begin{equation*}
        H(\mathbf{x},\mathbf{y},\mathbf{z},\mathbf{w},\mathbf{w}')=F(\mathbf{x},\mathbf{z},\mathbf{w})F(\mathbf{y},\mathbf{z},\mathbf{w}').
    \end{equation*}
    It is easy to see that $H^{[2a]}(\mathbf{x},\mathbf{y})=G(\mathbf{x},\mathbf{y})$ for all $\mathbf{x},\mathbf{y}\in D^a$. By \autoref{coro:vector}, $H^{[2a]}$ has a vector representation $(s_1,...,s_{2a})$ such that for all $\mathbf{x}\in  D^{[2a]}$, either $H^{[2a]}(\mathbf{x})=0$ or $H^{[2a]}(\mathbf{x})=s_1(x_1)\cdots s_{2a}(x_{2a})$. Let $s=s_1\otimes \cdots \otimes s_a$ and $t=s_{a+1}\otimes \cdots \otimes s_{2a}$. Then
    \begin{equation*}
        g=\begin{bmatrix}
            s(\mathbf{u}_1)t(\mathbf{u}_1) & s(\mathbf{u}_1)t(\mathbf{u}_2) \\
            s(\mathbf{u}_2)t(\mathbf{u}_1) & s(\mathbf{u}_2)t(\mathbf{u}_2)
        \end{bmatrix},
    \end{equation*}
    which is singular. A contradiction.
\end{proof}

Cai, Chen and Lu \cite{CCL2011} gave a criterion for \#CSP with non-negative weights:
\begin{lemma}
    Problem $\CSP(\mathcal{F})$ is in polynomial time if $\Gamma_\mathcal{F}$ is strongly rectangular and $\mathcal{F}$ is weakly balanced; otherwise it is $\sP$-hard.
\end{lemma}

Since balance implies strongly rectangularity, so does weak balance. Therefore, to determine the complexity of $\mathcal{F}$, we only need to decide whether $\mathcal{F}$ is of weak balance.

\section{Conclusion}

To determine the complexity of a problem $\Holant(\mathcal{F})$, the proofs of previous Holant dichotomies often start with a non-trivial function in $\mathcal{F}$. This works well for symmetric functions, but the structure of an asymmetric one can be very intricate. In \cite{CLX2011b}, we have already seen that asymmetry poses great challenges in arity reduction and gadget construction, even assuming the presence of all unary functions. In fact, similar difficulty arises on higher domains, where it is tough to obtain an explicit dichotomy. The \#CSP dichotomies over general domains \cite{DR2013,CCL2011,CC2012} are more abstract than those over the Boolean domain, but they offer great insights into sum-of-product computation. Inspired by them, we introduce the Block-rank-one condition for Holant problems, which leads to a clear classification. At the beginning of our work, we were not sure whether the condition is sufficient for tractability. Lemma \ref{lemma:balance-eq-block} and Lemma \ref{lemma:balance-unary} make it possible to absorb the results in \cite{DGJ2009} and reach the destination.

\section*{Acknowledgements}

This work was supported by the National Natural Science Foundation of China (Grants No. 61170299, 61370053 and 61572003).

\newpage
\bibliography{D:/texworks/mybib}

\end{document}